\documentclass[11pt, reqno]{amsart}
\usepackage{amsmath,amssymb,amsthm,amscd, amsfonts, mathrsfs}
\usepackage{amsaddr}

\usepackage{cases}

\usepackage[dvips]{epsfig}
\usepackage{epsf}
\usepackage{verbatim} 

\usepackage[bookmarksnumbered,pdfpagelabels=true,plainpages=false,colorlinks=true,
            linkcolor=black,citecolor=black,urlcolor=blue]{hyperref}

\theoremstyle{plain}
\newtheorem{theorem}{Theorem}[section]

\newtheorem{claim}[theorem]{Claim}

\theoremstyle{definition}
\newtheorem{remark}[theorem]{Remark}
\newtheorem{notation}[theorem]{Notation}
\newtheorem{definition}[theorem]{Definition}
\newtheorem{convention}[theorem]{Convention}

\numberwithin{equation}{section}

\newcommand{\cF}{\mathcal F}

\newcommand{\cI}{\mathcal I}

\newcommand{\cL}{\mathcal L}

\newcommand{\ccP}{\mathscr{P}}

\newcommand{\al}{\alpha}
\newcommand{\be}{\beta}
\newcommand{\ga}{\gamma}

\newcommand{\de}{\delta}
\newcommand{\ep}{\epsilon}
\newcommand{\la}{\lambda}
\newcommand{\La}{\Lambda}
\newcommand{\si}{\sigma}
\newcommand{\Si}{\Sigma}

\newcommand{\Om}{\Omega}

\newcommand{\RR}{\mathbb R}

\newcommand{\rar}{\rightarrow}

\newcommand{\p}{\parallel}

\newcommand{\K}{\mathscr{K}}

\newcommand{\mss}{\hspace{0.2cm}}

\newcommand{\msb}{\hspace{0.35cm}}

\newcommand{\bv}{\kappa}
\newcommand{\sv}{\vartheta}

\def\beq{\begin{equation}}
\def\ee{\end{equation}}

\allowdisplaybreaks

\title[Relativistic viscous fluids]{On the well-posedness of relativistic viscous fluids with
non-zero  vorticity}
\date{\today}

\author[Czubak]{Magdalena Czubak}
\address{Department of Mathematical Sciences \\ Binghamton University (SUNY) \\
Binghamton, NY, USA}
\email{czubak@math.binghamton.edu}
\thanks{Magdalena Czubak is partially supported by a grant from the Simons Foundation 246255.}

\author[Disconzi]{Marcelo M. Disconzi}
\address{Department of Mathematics\\
Vanderbilt University\\ Nashville, TN, USA}
\email{marcelo.disconzi@vanderbilt.edu}
\thanks{Marcelo M. Disconzi is partially supported by NSF grant 1305705.}

\begin{document}

\maketitle

\begin{abstract}
We study the problem of coupling Einstein's equations to 
a relativistic and physically well-motivated modification of the
Navier-Stokes equations. Under a technical condition 
for the vorticity, we prove existence and uniqueness in a 
suitable Gevrey class 
if the fluid's dynamic velocity has vanishing divergence,
and show that the solutions enjoy the finite propagation
speed property.
\end{abstract}

\tableofcontents

\section{Introduction and statement of the results. \label{intro}}

This paper improves the results of \cite{DisconziViscousFluidsNonlinearity}, where the formulation
of relativistic viscous fluids has been investigated by the 
second author. There, following Lichnerowicz 
\cite{LichnerowiczBookGR}, 
a physically well-motivated 
relativistic version of the Navier-Stokes equations has been proposed, and
well-posedness, in Gevrey spaces, 
of the corresponding Einstein-Navier-Stokes system
established under the assumption that the fluid is incompressible
(in a relativistic sense, see below) and irrotational. The goal 
of the present work is to remove the latter hypothesis, 
replacing it by a restriction on the initial data that allows
the vorticity to be non-zero. 

Finding the correct way of incorporating viscosity into 
General Relativity
is a longstanding problem\footnote{It is interesting to notice that
even the correct formulation of the non-relativistic Navier-Stokes
equations on a general Riemannian manifold does not seem to present
itself in a natural and obvious way, see \cite{ChanCzubakNavierStokes,EbinMarsden,TaylorPDE3}.} 
\cite{LichnerowiczBookGR, MisnerThorneWheeler, WeinbergCosmology},
one that has recently attracted attention due to its importance
in the study of heavily dense objects (as neutron stars),
 and models of the early
universe. See, for instance, 
\cite{Cooketal2003,
Disconzi_Kephart_Scherrer_2015,
Dosetal,
DosTsa,
Duez_review,
DuezetallEinsteinNavierStokes,
Herr_axially,
Herr_axially_shear,
Kreiss_et_al,
Lovelace_Duez_et_al,
MaartensDissipative,
Pahwa_FLRW,
Pal_et_al,
Pal_Reula_Rezzolla,
Piattella_et_al,
RezzollaZanottiBookRelHydro,
Saijo,
SorBran,
Weinberg_GR_book} 
 and references therein.
A thorough and more up-to-date discussion, including details
on the First and Second Order Theories mentioned below,
 can be found in \cite{RezzollaZanottiBookRelHydro}.

The main difficulty in formulating a theory of relativistic viscous fluids
 seems to stem from 
 the absence of 
a variational formulation for the classical, non-relativistic, Navier-Stokes
equations (see, however, \cite{EyinkActionPrinciple, YasueVariationalNavierStokes} for formalisms that allow
for a variational principle in some more general sense). 
Lacking such a formulation, one does not have a canonical
way of determining what the stress-energy tensor $T_{\al\be}$ ought to be in the 
context of General Relativity. Different proposals have been made in this
regard, but they have all led to either ill-posed equations, or to 
equations that imply the existence of
superluminal signals. This approach, where one couples Einstein's
equations to the Navier-Stokes equations via the introduction of a 
suitable stress-energy tensor is known as (relativistic) Standard Irreversible
Thermodynamics or First Order Theories. We remark that it 
consists in the traditional approach of coupling 
gravity to matter, one that has been successful for almost all other
matter fields \cite{ChoquetBruhatGRBook, HawkingEllisBook}.

The failure of First Order Theories in producing a consistent theory
of relativistic fluids led researches to devise a different approach, known 
as (relativistic) Rational 
Extended Irreversible Thermodynamics, or Second Order Theories, or yet 
Divergence-type formulation of Extended Irreversible 
Thermodynamics 
\cite{JouetallBook, MuellerRuggeriBook, RezzollaZanottiBookRelHydro}.
 
In such theories, one extends the space of variables of the model, and
the resulting equations are of hyperbolic character in several important
situations of physical interest, leading to equations that are well-posed,
with disturbances propagating at finite speed.
It is not at all clear, however, that the equations remain hyperbolic under all
physically realistic scenarios. In fact,
Rezzolla and Zanotti conclude their detailed discussion of relativistic viscous
fluids pointing out that ``the construction of a formulation that is cast in a 
divergence-type is not, \emph{per se}, sufficient to guarantee hyperbolicity"
\cite{RezzollaZanottiBookRelHydro}.
Furthermore, the plethora of models that comes out
of the extended thermodynamic approach suggests that it entails many
\emph{ad-hoc} features, in sharp contrast to the usually unique way of coupling
gravity to matter via the introduction of the stress-energy tensor of matter 
fields (when the latter
 is uniquely determined by a variational characterization).

These considerations 
suggest that it is  worthwhile to take a fresh look at the question
of whether there is a correct stress-energy tensor $T_{\al\be}$ that describes relativistic viscous
fluids, and that can be coupled to gravity in the traditional way, i.e.,
as in the Standard Irreversible Thermodynamics approach (see also the discussion
in section \ref{status}).
This idea is reinforced by the fact that  recent numerical advances in the 
modeling of rapidly rotating stars with shear viscosity employ 
the first order formalism \cite{DuezetallEinsteinNavierStokes}.

Consider the following stress-energy tensor for a viscous fluid:
\begin{align}
T_{\al\be} & = (p + \varrho ) u_\al u_\be - p g_{\al\be}
+ \bv \pi_{\al\be} \nabla_\mu C^\mu 
+ \sv\pi_\al^\rho \pi_\be^\mu (\nabla_\rho C_\mu + \nabla_\mu C_\rho),
\label{T_vis}
\end{align}
where
$p$ and $\varrho$ are respectively the
pressure and density of the fluid, $u$ is its four-velocity,
the bulk viscosity $\bv$ and the shear viscosity  $\sv$ are 
 non-positive constants\footnote{The coefficients of bulk and shear viscosity have a definite sign, the 
 choice of which depends on conventions. Sometimes $T_{\al \be}$
is written with a shear term $-\sv\pi_\al^\rho \pi_\be^\mu (\nabla_\rho C_\mu + \nabla_\mu C_\rho)$, $\sv >0$, which corresponds to having $\sv < 0$ in our formulation. 
While such sign differences are important when one explicitly
computes the values of physical observables, for the results here presented all that matters
is that $\sv \neq 0$. In physically realistic models, it is also the case that $\sv$ is 
not a constant, but a smooth function of the thermodynamic variables. Our result
can be extended to this case with minor changes in the proof, provided that
$\sv$ never vanishes, but we do not
include this here for brevity.}, 
$g$ is a Lorentzian metric\footnote{Our convention for the metric is $(+---)$.}
  and 
 $\pi_{\al\be} = g_{\al\be} - u_\al u_\be$.
$p$ and $\varrho$ are related by an equation known as equation of state, 
the choice of which depends on the nature of fluid. $C$ is known as
 the dynamic velocity 
  (also called the current of the fluid), and it is related to $u$ by
\begin{gather}
C_\al = F u_\al,
\label{dyn_vel}
\end{gather} 
where $F$ is the so-called index of the fluid. It is  defined as
\beq\label{indexF}
F = 1 + \epsilon + \frac{p}{r},
\ee
where $\ep\geq0$ is the specific internal energy and $r\geq 0$  is the rest mass density \cite{FriRenCauchy}. 
 The density $\varrho$ is related to the internal energy and the rest mass by
\[
\varrho=r(1+\ep),
\]
so that 
\beq\label{iF}
rF=\varrho + p.
\ee
The index of the fluid, $F$, and the dynamic velocity, $C$, have been introduced by
Lichnerowicz in his study of relativistic inviscid fluids
\cite{LichnerowiczBookGR, Lich_fluid_1, Lich_fluid_2, Lichnerowicz_MHD_book}.

Lichnerowicz was also the first one to write down the stress-energy tensor
 (\ref{T_vis}) \cite{LichnerowiczBookGR}, except that it  contained an extra term of 
 the form
$\sv \pi_{\al\be} u^\mu \partial_\mu F$.  This extra term was pointed out by
Lichnerowicz himself and later by Pichon \cite{PichonViscous}, to lead to an 
 indetermination in the computation of the pressure. Pichon proposed
 subtracting this term, which leads to (\ref{T_vis}).  
 See \cite{LichnerowiczBookGR,PichonViscous} for more background on \eqref{T_vis}. 
 The reader should notice that 
 (\ref{T_vis}) reduces to the stress-energy tensor of an ideal, i.e., inviscid, 
 fluid when $\bv = \sv = 0$. Indeed, this is just one of several natural 
 requirements that one would impose when looking for an appropriate
definition of a stress-energy tensor for a relativistic fluid with viscosity, 
see \cite{DisconziViscousFluidsNonlinearity}. We point out that Choquet-Bruhat has also proposed 
a stress-energy tensor similar to (\ref{T_vis}) \cite{ChoquetBruhatGRBook}. Her proposal
does not include the projection terms $\pi_{\al\be}$, and the viscous
terms are, therefore, linear in the velocity. We remark that yet another proposal
for a viscous relativistic stress-energy tensor appears in 
\cite{TempleViscous}.

Next, recall the first law of thermodynamics\footnote{See, for instance, \cite{DisconziRemarksEinsteinEuler,FriRenCauchy} for a review of the thermodynamic properties of 
relativistic fluids.} 
\beq\label{flt}
\theta ds=d\ep+pdv
\ee
where  $\theta$ is the absolute temperature, $s$ the specific entropy, and $v$ the specific volume.  We have $v=\frac 1r$ \cite{FriRenCauchy}, 
so by \eqref{indexF}, \eqref{flt} can be written as
\beq\label{flt2}
\theta ds=dF-\frac 1r dp.
\ee

A fluid with stress-energy tensor (\ref{T_vis}) is said to be incompressible if 
\beq\label{incompressible_def}
\nabla_\mu C^\mu = 0.
\ee  
\begin{remark}
Here we follow the literature (e.g., \cite{ChoquetBruhatGRBook, Lichnerowicz_MHD_book})
and call incompressible a fluid satisfying (\ref{incompressible_def}). We stress, however, that this terminology
is based more on an analogy with Newtonian physics (where incompressible fluids are characterized 
by vanishing divergence) than on actual physical properties of fluids, in that (\ref{incompressible_def})
does not imply incompressibility in the exact sense of the word.
Pseudo-incompressible would probably be a better terminology, but it is not clear if adopting a different terminology
than what is used in the literature would not cause more confusion. 
\label{remark_incompressibility}
\end{remark}

Then,  by \eqref{iF},  $T_{\al\be}$ 
becomes
\begin{align}
T_{\al\be}  = rF u_\al u_\be - p g_{\al\be}
+ \sv\pi_\al^\rho \pi_\be^\mu (\nabla_\rho C_\mu + \nabla_\mu C_\rho).
\label{T_vis_i}
\end{align}
Moreover, because $\pi^\mu_\beta u_\mu=0$, we can rewrite \eqref{T_vis_i} as
\beq\label{Tab}
T_{\al\be}  = rF u_\al u_\be - p g_{\al\be}
+ F\sv\pi_\al^\rho \pi_\be^\mu (\nabla_\rho u_\mu + \nabla_\mu u_\rho).
\ee

Finally, we define the vorticity tensor by
\begin{gather}
\Om_{\al\be } = \nabla_\al C_\be - \nabla_\be C_\al 
\equiv \partial_\al C_\be - \partial_\be C_\al.
\label{vorticity}
\end{gather}
A fluid is called irrotational if $\Om = 0$. Notice that $\Om$ is anti-symmetric,
so it has only six independent components.

We follow the standard approach of assuming that only two
of the thermodynamic quantities are independent with the question of which ones left as a matter of choice.  The other quantities are then
determined by the first law of thermodynamics and an equation of state.  The equation of state 
depends on the nature of the fluid, and physically, the relations between the thermodynamic quantities 
should be invertible\footnote{Upon making such assumptions, we are restricting ourselves to fluids in a single phase and
ruling out the possibility of phase transitions.}.  Here we shall assume that 
$r$ and $s$ are independent and postulate an equation of state 
of the form
\begin{gather}
\varrho = \ccP(r,s).
\label{eq_state_r_s}
\end{gather}
It follows that $p = p(r,s)$, 
$\theta=\theta(r,s)$, $\epsilon=\epsilon(r,s)$, 
and $F=F(r,s)$ are known if $r$ and $s$ are. We note that later on it will be more convenient to treat $s$ and $F$
as independent variables. Then the equation
of state will be given by $r = r(F,s)$.

On physical grounds, one has that $F > 0$.
This allows to restrict to positive values when treating $F$ as
an independent variable.
In this situation, the following
condition will be assumed to hold:
\begin{gather}
\frac{\partial r}{\partial F} \geq  \frac{r}{F},
\label{sound_speed_condition}
\end{gather} 
in particular $\frac{\partial r}{\partial F} > 0$ if
$r > 0$. 
Condition (\ref{sound_speed_condition}) expresses the statement that sound waves in an ideal fluid
travel at most at the speed of light.  This condition has to be satisfied if we want to recover the stress-energy
tensor of an ideal fluid when $\bv = \sv = 0$ \cite{Lichnerowicz_MHD_book}
We suppose that the equation of state is such that the temperature
satisfies
\begin{align}
\begin{split}
\theta(r,s) > 0 \text{ if } r > 0, s \geq 0, \\
\theta(F,s) > 0 \text{ if } s \geq 0, F > 0,
\end{split}
\label{temperature}
\end{align}
expressing the positivity of the temperature regardless of  the choice
of independent variables.

The full system of equations derived from coupling Einstein's equations
to (\ref{T_vis_i}) is rather complicated, see 
\cite{DisconziViscousFluidsNonlinearity}. Thus we consider, besides incompressibility, one further
simplifying assumption, namely, we investigate the sub-class of solutions 
for which the vorticity evolves according to 
\begin{gather}
\cL_C \Om_{\al \be}
+ q u^\mu \nabla_\mu \nabla_\al C_\be - q
u^\mu \nabla_\mu \nabla_\be C_\al +
\partial_\al(\theta F) \partial_\be s - \partial_\be(\theta F) \partial_\al s
= \cF_\sv.
\label{evolution_vorticity}
\end{gather}
Here, $\cL_C$ denotes the Lie derivative in the direction of $C$, 
$\cF_\sv$ is a smooth function
 of $\Om$, $g$ and its derivatives up to second
order, $C$ and $u$ and their derivatives up to first order.
$q$ is a constant, and $\cF_\sv$ and $q$ 
may also depend on the parameter $\sv$. $\cF_\sv$ 
and $q$ dictate, to a certain extent, which quantities are considered
relevant in some particular model, and, therefore, are chosen according
to the problem one wishes to study (see section \ref{vorticity_in_rel}). 
We discuss the restrictions imposed by 
(\ref{evolution_vorticity}) in section \ref{restriction_vorticity}.
It should be noticed that one must have $q=0$ and that $\cF_\sv$
cannot be chosen freely if $\sv = 0$ (although $\sv = 0$ will not be treated here).

The starting point is the following system of equations: Einstein's equations coupled to 
(\ref{T_vis}) and supplemented by (\ref{incompressible_def}) and 
(\ref{evolution_vorticity}), reading
\begin{subnumcases}{\label{original}}
R_{\al\be} - \frac{1}{2}R g_{\al\be} + \La g_{\al\be}= \K T_{\al\be},
\label{original_Einstein} \\
\nabla^\al T_{\al\be} = 0,
\label{original_conservation} \\
\nabla_\al( r u^\al )  = 0, 
\label{original_mass_conservation} \\
\nabla_\mu C^\mu = 0,
\label{original_incompressible} \\
\cL_C \Om_{\al\be} +q u^\mu \nabla_\mu \nabla_\al C_\be - q
u^\mu \nabla_\mu \nabla_\be C_\al = \widetilde{\cF}_\sv, 
\label{original_vorticity}
\\
u^\al u_\al = 1.
\label{original_normalization}
\end{subnumcases}
where $R_{\al\be}$ and $R$ are the Ricci and scalar curvature
of the metric $g$,  $\K$ is a coupling constant, and $\La$ is 
the cosmological constant\footnote{Our results do hold irrespective
of the value of $\La$.}.  We recall that (\ref{original_conservation}) is in fact a consequence 
of (\ref{original_Einstein}) in view of the Bianchi identities, but it is customary to list it along with the other equations.
In the sequel we set $\K$ to $1$. The equation 
(\ref{original_mass_conservation}) says
the mass is locally conserved along the flow lines, and (\ref{original_normalization}) is the standard
normalization condition on the velocity of a relativistic fluid.  
In general, without 
(\ref{original_mass_conservation}), the motion of the fluid is underdetermined.
Equation (\ref{original_vorticity}) is simply (\ref{evolution_vorticity}), 
with $\widetilde{\cF}_\sv$ defined in the obvious fashion.
  A useful consequence of \eqref{original_normalization} often used in computations is
\beq\label{useful}
u^\al \nabla_\be u_\al=0.
 \ee
The unknowns are the metric $g$, the fluid velocity $u$,  the specific entropy $s$, and the rest mass density $r$, where  $s$ and $r$ are non-negative
real valued functions. We suppose that we are also given a
smooth function $\ccP: \RR_+ \times \RR_+ \rar \RR$ that gives the 
 equation of state (\ref{eq_state_r_s}), with the other thermodynamic quantities  then given as functions of $s$ and $r$ as discussed above.

\begin{definition}
System (\ref{original}) with $T_{\al\be}$ 
given by (\ref{T_vis}) will be called the
incompressible Einstein-Navier-Stokes
system.
\end{definition}

\begin{remark}
Here we recall once more that the terminology ``incompressible fluid" is a bit misleading, see remark \ref{remark_incompressibility}.
\end{remark}

\noindent \textbf{Assumption.} 
We shall assume for the rest of the text that $\sv \neq 0$. \\

An initial data set for the Einstein-Navier-Stokes system consists of the following:
\begin{itemize}
\item a three-dimensional  manifold $\Si$, 
\item a Riemannian metric $g_0$ (with our conventions this metric is negative definite)
\item   a symmetric two-tensor $\kappa$, 
\item a real valued non-negative function $s_0$, 
\item a real valued non-negative function $r_0$,
\item a vector field $v$.
\end{itemize}
The last five quantities are defined on $\Si$. As it is well-known, these data must satisfy the constraint equations.  In a coordinate system with $\partial_0$ transversal and
 $\partial_i$, $i=1,2,3$, tangent to $\Si$ the constraint equations are given by
\begin{gather}
S_{\al 0 } =  T_{\al 0 },
\label{constraints}
\end{gather}
where 
$S_{\al\be} = R_{\al\be} - \frac{1}{2}R g_{\al\be} + \La g_{\al \be}$
is the Einstein tensor. In our case, it is not enough that
the initial data satisfies (\ref{constraints}). We also need
the compatibility conditions obtained upon restriction 
of (\ref{original_vorticity}) to the initial hypersurface, 
since initial data for $\Om$ and $C$ are derived from 
$g_0$, $\kappa$, $s_0$, $r_0$, and $v$ (see \cite{DisconziViscousFluidsNonlinearity}).
By definition, an initial data set 
always satisfies the constraints and compatibility 
conditions. While the construction of initial data 
for the Einstein-Navier-Stokes equations is an important task,
here our primary interest is on the evolution problem, and as
such we shall take the standard approach of assuming the initial data
as given (see the discussion in section \ref{restriction_vorticity}).

We are now ready to state the main result. We refer the 
reader to the standard literature in General Relativity for the terminology
employed in Theorem \ref{main_theorem}. We remind the reader of the 
definition of Gevrey spaces $\ga^{m, (\si)}$ in Section \ref{main_proof}, referring to
 references
 \cite{LerayOhyaLinear, LerayOhyaNonlinear,RodinoGevreyBook} for more details. 

\begin{theorem}
Let $\cI = (\Si, g_0, \kappa,v, s_0,  r_0)$ be an initial data set for the 
 incompressible Einstein-Navier-Stokes system  (\ref{original}) , 
 with $\Si$ compact,  $s_0 > 0$, $r_0 > 0$, and an equation of state $\ccP$
such that (\ref{sound_speed_condition}) and (\ref{temperature}) 
are satisfied. Let $\cF_\sv$ be a given smooth function of 
$\Om$, $g$ and its derivatives up to second
order, $C$ and $u$ and their derivatives up to first order, and assume that
$q > 0$.
Assume that the initial data is in 
$\ga^{(\si)}(\Si)$ for some $1 \leq \si <\frac{24}{23}$.
Then there exist a space-time $(M,g)$
that is a development of $\cI$,
 real valued  functions $s > 0$ 
and $r > 0$
defined on $M$, and a vector field $u$, 
such that  
$g \in \ga^{3, (\si)}(M)$,
$u \in \ga^{2, (\si)}(M)$, 
$s \in \ga^{2, (\si)}(M)$, 
$r \in \ga^{2, (\si)}(M)$, and  
$(g,u, s, r)$ satisfy the  incompressible 
Einstein-Navier-Stokes system  in $M$.

Furthermore, this solution satisfies the geometric uniqueness
and domain of dependence properties, in the following sense. 
Let $\cI^\prime = (\Si^\prime, g_0^\prime, \kappa^\prime, v^\prime, s_0^\prime,  r_0^\prime)$
be another initial data set, also with equation of state $\ccP$,
with corresponding development $(M^\prime,g^\prime)$ and 
solution $(g^\prime,u^\prime, s^\prime, r^\prime)$ 
of  the incompressible Einstein-Navier-Stokes equations
in $M^\prime$.
Assume  that there exists a diffeomorphism
 between $S \subset \Si$ and $S^\prime \subset \Si^\prime$ that carries
$\left. \cI\right|_S$ onto $\left. \cI^\prime \right|_{S^\prime}$, where $S$ and $S^\prime$
are, respectively, domains in $\Si$ and $\Si^\prime$. Then there exists
a diffeomorphism between $D_g(S) \subset M$
and  $D_{g^\prime}(S^\prime) \subset M^\prime$ carrying
$(g,u, s, r)$ onto
$(g^\prime,u^\prime, s^\prime, r^\prime)$, where $D_g(S)$ denotes
the future domain of dependence of $S$ in the metric $g$; in particular $D_g(S)$ and 
$D_{g^\prime}(S^\prime)$ are isometric.
\label{main_theorem}
\end{theorem}
\begin{remark}
Under the further assumption that the fluid is irrotational, 
Theorem (\ref{main_theorem}) was proved by the second author in
\cite{DisconziViscousFluidsNonlinearity}, where a  better regularity result than
in theorem \ref{main_theorem}, namely, $\si < 2$, was obtained.
\end{remark}
\begin{remark}
The space-time $M$ is diffeomorphic to $\Si \times [0,T]$ for some $T > 0$, 
and to $\Si \times [0,\widetilde{T})$ for some $\widetilde{T} > T$
 if we require it to be a maximal Cauchy 
development.
\end{remark}
\begin{remark}
 The compactness of $\Si$
is not absolutely necessary due to the domain of dependence property.  However, in the case of a 
non-compact $\Si$ without asymptotic conditions on the initial data, $M$ may not contain any Cauchy surface other than 
$\Si$ itself. 
\end{remark}
\begin{remark}
The hypotheses $s_0>0$ and $r_0 > 0$ guarantee, by continuity, 
the positivity of $s$ and $r$ in the neighborhood of $\Si$, as stated in the theorem.
The assumption $s_0 > 0$ could be weakened to $s_0 \geq 0$, but in this case 
the non-negativity of $s$ in $M$ would have to be derived from the equations of
motion, a task we avoid for brevity. On the other hand, allowing $r_0$ to vanish would cause severe
difficulties.  In fact, the well-posedness of the corresponding problem is largely 
open even in the case of an ideal fluid \cite{FriRenCauchy}.
\end{remark}

In the following, we adopt:

\begin{convention}Greek indices run from $0$ to $3$ and Latin indices
from $1$ to $3$. 
\end{convention}

\section{Discussion on the hypotheses and the thesis of Theorem \ref{main_theorem}}
 
In this section we comment on the relevance 
of Theorem \ref{main_theorem} for the study of relativistic fluids
with viscosity. We highlight the restrictions imposed by (\ref{evolution_vorticity})
and the regularity of solutions,
make some remarks regarding the physical content of the Theorem, and discuss 
how this work fits within the broader context of relativistic viscous fluids,
making some general remarks about (\ref{T_vis}) and its particular case (\ref{T_vis_i})
along the way.
Readers interested solely in the proof of Theorem \ref{main_theorem} may skip this section.

\subsection{The evolution of the vorticity.\label{restriction_vorticity}}
Let us start with the evolution condition imposed on $\Om$, i.e.,
equation (\ref{evolution_vorticity}) or, equivalently, (\ref{original_vorticity}). 
In its full-generality, 
the incompressible Einstein-Navier-Stokes system consists
of equations (\ref{original_Einstein})-(\ref{original_incompressible}) and
(\ref{original_normalization}), i.e., (\ref{original}) without
 (\ref{original_vorticity}).
As such a system is rather complicated, we have imposed
(\ref{evolution_vorticity}) which, of course, is ultimately
a restriction on the unknowns $(g,u,s,r)$. From
the point of view of 
(\ref{original_Einstein})-(\ref{original_incompressible}), equation  
(\ref{original_vorticity}) should be understood as a constraint, 
in the following sense. An initial data set 
yielding a solution to 
(\ref{original_Einstein})-(\ref{original_incompressible}) (plus
(\ref{original_normalization})) 
will also give a solution to (\ref{original}) only if 
further relations among the initial data hold. 
Indeed, arguing as in  \cite{DisconziViscousFluidsNonlinearity},
one determines, from the original Cauchy data and Einstein's equations, 
the values of $\partial^2g$, $\partial s$, $\partial u$, $\Om$, and $\partial^2 C$
(as well as the corresponding lower order derivatives) on the initial hypersurface $\{ t = 0 \}$,
obtaining 
a relation of the form  $\partial_0 \Om_{\al\be} = W_{\al \be}$
 on $\{ t = 0 \}$, where
$W_{\al \be}$ is a function of the Cauchy data.
From (\ref{evolution_vorticity}), one also obtains 
a relation $\partial_0 \Om_{\al\be} = Z_{\al \be}$
 on $\{ t = 0 \}$, with  $Z_{\al \be}$  a function of the Cauchy data.
Therefore the initial data must be such data $W_{\al\be} = Z_{\al\be}$,
and hence Theorem \ref{main_theorem} is ultimately a result under restrictions on the initial data, namely,
the initial data ought to satisfy 
compatibility 
conditions imposed by (\ref{original_vorticity}), as mentioned 
earlier\footnote{We notice that a similar, albeit notably simpler, situation happens to perfect 
fluids with zero vorticity. In that case, the equation for the vorticity is given by
$\cL_C \Om = 0$, where $\cL$ is the Lie derivative. This is not compatible with 
the equations derived from the divergence of the stress-energy tensor unless the initial data
is such that $\Om = 0$ on the $\{ t = 0\}$ slice.}. 

How large is the class of initial conditions
satisfying the above restrictions is by no means an unimportant question,
but one that is, at this point, premature,
since we do not even know whether the system (\ref{original_Einstein})-(\ref{original_incompressible}) and
(\ref{original_normalization}) has any solution at all
outside the class of analytic functions. 
While this leaves open the question of how general 
Theorem \ref{main_theorem} is,
it is consistent with some expected physical applications as
discussed in section \ref{vorticity_in_rel}. 
Such restrictions notwithstanding, we make two important remarks.

First, one could, in principle, consider the case of zero vorticity, with the function 
$\cF_\sv$ chosen so that (\ref{evolution_vorticity}) becomes an identity (notice that our results
do not rely on any specific form of $\cF_\vartheta$, except for the dependence
on the number of derivatives of its arguments). In this case, the constraints 
reduce to those that have to be imposed on the initial data when $\Om = 0$. One could
say then that our theorem reproves the earlier result 
\cite{DisconziViscousFluidsNonlinearity}. While 
obviously this is not our goal here, it at least shows the set of appropriate
initial data to be non-empty. 
The interesting question is whether the set of initial data satisfying the constraints 
and compatibility conditions is non-empty
modulo zero vorticity. This is will be addressed in a future work.

Second, in light of how little is known about viscosity in General Relativity
(see section \ref{status}), our conditional result should be view
as evidence that further investigation of Lichnerowicz's proposal (\ref{T_vis})
is a worthwhile line of inquiry. In this regard,
it is illustrative to point out that there are other situations in General 
Relativity where the evolution problem
is investigated without a decisive answer to the question of solvability of the constraints,
but this has never stopped the community to make conditional statements regarding
Einstein's equations. 
One such situation, for example, is the study of vacuum Einstein's equations with low regularity.
As discussed, for instance, in \cite{HolstMeier},
the ``rough solution theory of the
constraints has in fact lagged behind that of the evolution problem." 
For instance, well-posedness of the vacuum evolution problem for
data $(g,\kappa$) in $H^s\times H^{s-1}$, $s > \frac{5}{2}$, had been known
since 1977 \cite{Hughes_Kato_Marsden_quasi_linear_1977}. 
However, it was not until 2004 that solutions to the constraint
equations in this regularity class could be constructed \cite{CB_constraints_2004}.
Hence, for 27 years it was not known if the classical result \cite{Hughes_Kato_Marsden_quasi_linear_1977}
was not empty modulo large 
values of $s$ (for which \cite{Hughes_Kato_Marsden_quasi_linear_1977} would simply reproduce
earlier known results).

\subsection{Regularity of solutions.}
We work in the Gevrey class, because the equations
we derive form a Leray-Ohya system
(see \cite{LerayOhyaNonlinear}), which, in general, is not well-posed 
in Sobolev spaces.  Gevrey spaces have become an important tool in analyzing the 
equations of Fluid Dynamics, especially when viscosity is present 
(see, e.g., \cite{TadmorBesovGevrey,  TitiGevreyNavier, TitiGevreyParabolic, TemamGevrey, RodinoGevreyBook} and references therein). Hence, it is sensible
that such spaces might play a role in the case of relativistic
 viscous fluids as well. 
Furthermore, Gevrey spaces are not completely foreign to the study of Einstein's 
equations:
in some relevant circumstances, the equations of ideal magneto-hydrodynamics appear 
to 
have been shown to be well-posed only in the 
Gevrey class \cite{ChoquetBruhatGRBook,FriRenCauchy}\footnote{Although it is very likely that
the formulation of \cite{AnilePennisiMHD} would carry over, with almost no modifications, to the 
coupling
with Einstein's equations. A proof of this statement, however, does not
seem to be available in the literature.}. On the other hand,
the overwhelming success of Sobolev space techniques in the investigation of the 
Cauchy problem
for Einstein's equations\footnote{The literature on this topic is too vast; see, 
e.g., the monographs \cite{ChoquetBruhatGRBook, RingstromCauchyBook}.} 
almost demands that we employ Sobolev spaces in the study of the 
evolution
problem. Moreover, in order to eventually settle the question of 
whether (\ref{T_vis}) can give a physically satisfactory description
of relativistic viscous phenomena, we have to be able to explicitly compute several 
physical
observables. For this, one has to solve the equations numerically, which, in turn,
requires that the equations be well-posed in some function space characterized by 
a finite number of derivatives.

Unfortunately, currently Gevrey regularity seems to be the best one can do for (\ref{T_vis}),
as the corresponding equations of motion do not seem to be amenable to known Sobolev-type techniques.
We remark, however, that simply establishing causality of the equations of motion is already 
a step forward in light of the long history of non-causal theories of relativistic viscous fluids
\cite{RezzollaZanottiBookRelHydro}.

\subsection{The status of viscosity in relativity\label{status} and Lichnerowicz's proposal}

In spite of the restriction on the initial data due to (\ref{evolution_vorticity}),
the severity of which we acknowledge is yet
to be understood, and on the regularity class of solutions, 
one should not overlook the 
conclusion of Theorem \ref{main_theorem}: it is possible, employing
the traditional Standard Irreversible Thermodynamics, to obtain 
a description of relativistic viscous fluids that is well-posed and 
does not 
exhibit faster than light signals. In this regard, we
remind the reader once more that we are attempting
a new look at this problem through a first order formalism.
Hence, it is all but unreasonable to start off with conditions that
render the problem tractable with current mathematical technology.
The first attempt in this direction \cite{DisconziViscousFluidsNonlinearity} dealt with irrotational
fluids. Here, we considered a less dramatic condition on the vorticity,
namely, (\ref{evolution_vorticity}), which seems to be compatible
with some physical applications (see section \ref{vorticity_in_rel}).
The message conveyed by this is that,
while it is wide open whether a full existence result for the incompressible Einstein-Navier-Stokes may be within reach without restrictions on the vorticity,
one can still prove well-posedness results under interesting scenarios.

Another restriction in our Theorem that one would like to remove is the
incompressibility hypothesis, not only for the sake of mathematical generality,
but also because relativistic systems many times exhibit sound waves that
propagate at sub-luminal speeds. This is the subject of ongoing investigations.

In order to put all of the above in perspective, we give rather brief 
overview of what is currently known about viscosity in relativity.

The Mueller-Israel-Stewart (MIS) theory
\cite{MIS-2, MIS-3, MIS-5, MIS-6, MIS-1, MIS-4}
is probably the best accepted theory of relativistic viscous phenomena. It consists of 
a systematic application of the ideas of Relativistic Extended Irreversible Thermodynamics
\cite{JouetallBook, MuellerRuggeriBook}.
The linearization about equilibrium states of the MIS theory
has been shown to be causal \cite{Hiscock_Lindblom_stability_1983}.
The non-linear theory, however, is also plagued
with non-causality behavior \cite{Hiscock_Lindblom_pathologies_1988}.
To be fair,  such loss of causality is known to happen under extreme physical conditions 
unlikely to be met by most realistic systems.
More precisely, in \cite{Hiscock_Lindblom_pathologies_1988}
the authors investigate 
the relatively simple case 
where only heat conduction is 
present, so that the bulk and shear viscosity are zero, and under the assumption of planar symmetry. 
 Under
these assumptions, it is shown in \cite{Hiscock_Lindblom_pathologies_1988} that the equations
of motion are causal under a restriction on the values of the heat-flux, and non-causal 
if such a restriction is violated\footnote{In passing, one should note that the MIS is sometimes referred to 
as ``causal dissipative relativistic theory," but strictly speaking that
is, in view of what has been said, a misnomer.}.
In contrast, Theorem \ref{main_theorem}, as well as  \cite{DisconziViscousFluidsNonlinearity},
makes no symmetry or near-equilibrium assumption, and treats the full non-linear system,
albeit it assumes stiffness and stringent restrictions on the initial data (or irrotationality
in the case of \cite{DisconziViscousFluidsNonlinearity}). It is important to notice, 
however, that from the point
of view of causality, such results treat precisely the most ``dangerous" scenario, i.e., they
include  
the term
\begin{gather}
 \pi_\al^\rho \pi_\be^\mu 
(\nabla_\rho C_\mu + \nabla_\mu C_\rho),
\nonumber
\end{gather}
which leads to multiple characteristic due to the presence 
of the projections $\pi_\al^\rho \pi_\be^\mu$. The causality obtained in 
\cite{Hiscock_Lindblom_pathologies_1988}, on the other hand, is restricted to the case
when the viscous part of the stress-energy tensor contains only the heat flow; in particular, such 
projection terms are absent.

We also point out that, to the best of our knowledge, 
the aforementioned causality and well-posedness results of the MIS 
theory \cite{Hiscock_Lindblom_stability_1983, Hiscock_Lindblom_pathologies_1988} do not include
coupling to Einstein's equations, i.e., they consider the fluid equations in a fixed background
(except for some very simple situations such as FRW cosmologies \cite{MaartensDissipative}),
whereas Theorem \ref{main_theorem} and  \cite{DisconziViscousFluidsNonlinearity} do 
treat the full Einstein-fluid system.

Another interesting feature of Theorem \ref{main_theorem} is that it circumvents the
instability results of Hiscock and Lindblom
\cite{Hiscock_Lindblom_instability_1985}. 
In fact, formally the equations that we study here
correspond to the case  $\kappa = \si=0$ in \cite{Hiscock_Lindblom_instability_1985}.
Equations (\ref{incompressible_def}) and (\ref{evolution_vorticity}), however, 
further constrain the evolution of perturbations  
(compare with equations (41) in \cite{Hiscock_Lindblom_instability_1985}).
On the other hand, if condition (\ref{incompressible_def}) is dropped, then
the term $\nabla_\mu C^\mu$ that contributes 
to the viscous part in (\ref{T_vis}) will depend on derivatives of the termodynamic
variables along the flow, a case not covered under the assumptions of
\cite{Hiscock_Lindblom_instability_1985}.

One important question about theories based on (\ref{T_vis}) is whether natural physical requirements
are satisfied. One such requirement is that entropy production be non-negative. It is not difficult
to see that, at least for the case investigated here, namely, when (\ref{T_vis}) reduces to
(\ref{T_vis_i}), this is the case.  To see this, 
one first uses $u^\be \nabla^\al T_{\al\be} = 0$ and the first law of thermodynamics to derive
\begin{gather}
\theta r u^\al \partial_\al s = \sv F ( \nabla^\mu u^\nu  \nabla_\mu u_\nu 
+ \nabla^\mu u^\nu  \nabla_\nu u_\mu - u^\mu \nabla_\mu u^\al u^\nu \nabla_\nu u_\al).
\nonumber
\end{gather}
On the other hand, direct computation gives
\begin{gather}
\Si^{\al\be} \Si_{\al \be} = 2F^2
( \nabla^\mu u^\nu  \nabla_\mu u_\nu 
+ \nabla^\mu u^\nu  \nabla_\nu u_\mu - u^\mu \nabla_\mu u^\al u^\nu \nabla_\nu u_\al),
\nonumber
\end{gather}
thus
\begin{gather}
\theta r u^\al \partial_\al s =  \frac{\sv}{2F} \Si^{\al\be} \Si_{\al \be} \geq 0,
\nonumber
\end{gather}
since $\Si^{\al\be} \Si_{\al \be} \leq 0$\footnote{Recall our convention for the metric.} and $\sv \leq 0$.
A detailed study of physical consistency of models based on (\ref{T_vis}) appears in 
\cite{DisconziKephartScherrerNew}
(see \cite{Disconzi_Kephart_Scherrer_2015}, however, for some further results concerning physical 
properties of (\ref{T_vis})).

In summarty, it is fair to say that despite considerable progress since the original work of Eckart
\cite{EckartViscous}, the description of relativistic viscous phenomena still presents
many challenges. These remarks are not intended to claim that Lichnerowicz's proposal is better than 
the more studied MIS theory or should be favored over any other theory, 
but rather to highlight how little is known about viscosity in General
Relativity, which makes, in our opinion, attempts at different approaches, such as those based
on (\ref{T_vis}), welcome.

\subsection{Vorticity in relativistic fluids and some physical considerations\label{vorticity_in_rel}}
In the case of inviscid fluids, the vorticity obeys \cite{Lich_fluid_1, Lichnerowicz_MHD_book}
\begin{gather}
\cL_C \Om_{\al \be} + 
\partial_\al(\theta F) \partial_\be s - \partial_\be(\theta F) \partial_\al s
= 0.
\label{vorticity_ideal}
\end{gather}
This equation, sometimes called the Lichnerowicz, or Carter-Lichnerowicz
equation, plays an important role in the study of inviscid relativistic fluids,
 and  generalizations have also been
employed in formulations with viscosity 
\cite{AndCom, ChrisMalik, Dosetal, DosTsa,
GourgIntro, Herr_axially_shear, Lewis, SorBran, TsagasBarrow}. 

Equation (\ref{evolution_vorticity}) reduces to (\ref{vorticity_ideal})
when $\cF_\sv = 0 = q$; thus, in particular,  we see that
in physically relevant models, $q$ and $\cF_\sv$ vanish when viscosity
is absent. Many of the modifications of (\ref{vorticity_ideal})
that include viscosity tend to occur in the context of very specific models,
where the equations are simple as compared to the ones here investigated
(for instance, perturbations of an FRW model). These can generally 
be accommodated by (\ref{evolution_vorticity}) with a suitable choice of 
$\cF_\sv$; see, for example, 
\cite{AndCom, ChrisMalik, Dosetal, DosTsa, Herr_axially, Herr_axially_shear, SorBran} and references therein. In more general terms, (\ref{evolution_vorticity})
seems natural when one considers applications with small viscosity,
in that the evolution of $\Om$ should, to a certain degree,
resemble that of an ideal fluid. We also point out that
 (\ref{evolution_vorticity}) is consistent with  standard cosmology (with no 
viscosity), in that, in such scenarios,  $\Om$ decays with the Hubble expansion, 
being, as a consequence, ignored in many circumstances\footnote{Although
vorticity may play an important role in early stages of the Universe.
See, for example, 
\cite{ChrisMalik} and references therein.}. Hence, one may, again,
suspect that in these cases $\Om$ will be governed by an equation
similar to that of ideal fluids, since (\ref{vorticity_ideal})
enjoys the property of preserving zero vorticity.
Our choice (\ref{evolution_vorticity})
is a compromise between the previous considerations and an 
algebraic form for which properties of hyperbolic polynomials, necessary for our techniques, hold true.
See section \ref{main_proof}.
 
\section{A new system of equations.\label{new_system_section}}
Here we derive a different system of equations, whose existence of solutions
implies Theorem \ref{main_theorem}. Thus, for the rest of this section, we assume we have a
 sufficiently differentiable solution to (\ref{original}). In particular, 
 in light of (\ref{dyn_vel}) and (\ref{original_normalization}), 
 one has
\begin{gather}
F = \sqrt{ C^\mu C_\mu},
\nonumber
\end{gather}
so that $F$ can be viewed as a function of $g$ and $C$. 
 
\begin{convention} Unless stated otherwise, we shall assume from now 
on to be working in harmonic (or wave) coordinates.
\end{convention}

\begin{notation}
Below, $B$, 
with indices attached when necessary, is used to denote expressions, where the 
maximum number of the derivatives of the variables $g$, $s$, $u$, $\Om$, $F$, 
and $C$ is
indicated in the arguments. For instance, $B(\partial g, \partial^2 s)$ indicates an 
expression depending on at most first derivatives of $g$ and second derivatives of $s$. The expression represented by $B$ can vary from equation to equation. 
\end{notation}

\subsection{Equation for $g$}
By taking the trace of \eqref{original_Einstein}  we get
\[
T:=g^{\al\be}T_{\al\be}=-R + 4\La,
\]
so we can rewrite Einstein's equation as
\beq\label{ens1}
R_{\al\be}=T_{\al\be}-\frac 12 Tg_{\al\be} + \La g_{\al\be}.
\ee
Next, recall that in harmonic coordinates, the Ricci curvature can be written as 
\beq\label{ens2}
R_{\al\be}=\frac 12 g^{\mu\nu}\partial_{\mu\nu}g_{\al\be}+B_{\al\be}(\partial g).
\ee
From \eqref{Tab} we also have
\begin{align}
T_{\al\be}&= rF u_\al u_\be - p g_{\al\be}
+F\sv\pi_\al^\rho \pi_\be^\mu (\nabla_\rho u_\mu + \nabla_\mu u_\rho)\nonumber\\
&=B_{\al\be}(\partial g, s, u, C)+ \sv\sqrt{C^\nu C_{\nu}}\pi_\al^\rho \pi_\be^\mu (\partial_\rho u_\mu + \partial_\mu u_\rho)\label{ens3}.
\end{align}
Hence
\beq\label{ens4}
T=rF-4p+\sv F\pi^{\rho\mu}(\nabla_\rho u_{\mu}+\nabla_{\mu}u_\rho)=\sv\sqrt{C^\nu C_{\nu}}\pi^{\rho\mu} (\partial_\rho u_\mu + \partial_\mu u_\rho)+B(\partial g, s, u, C).
\ee
Inserting \eqref{ens2}, \eqref{ens3}, and \eqref{ens4} into \eqref{ens1} we obtain
the following equation for $g$
\beq\label{ens5}
 g^{\mu\nu}\partial_{\mu\nu}g_{\al\be}-\sv\sqrt{C^\nu C_{\nu}}\left(2\pi_\al^\rho \pi_\be^\mu (\partial_\rho u_\mu + \partial_\mu u_\rho) -\pi^{\rho\mu} (\partial_\rho u_\mu + \partial_\mu u_\rho)g_{\al\be}\right)+B_{\al\be}(\partial g, s, u, C)    =0.
\ee

\subsection{Equation for $s$}
From \eqref{original_conservation}, \eqref{original_mass_conservation}, \eqref{flt2} by considering $u^\be\nabla^\al T_{\al\be}$ with $T_{\al\be}$ as in \eqref{Tab}, we obtain
\[
r\theta u^\al \partial_\al s-\sv F(\nabla_\mu u_\nu +\nabla_\nu u_\mu)\pi^{\al \mu}\nabla_\al u^\nu=0,
\]
where we also used \eqref{useful} and that $u^\be\pi^\mu_\beta=0$.

To obtain the desired quasi-linear structure we apply $u^\si\nabla_\si$ to the equation.  This results in 
\beq\label{s1}
\begin{split}
u^\si u^\al \partial_{\al\si} s &- \sv\frac{\sqrt{C^\rho C_\rho}}{\theta r} \pi^{\al\mu}u^\si \partial_\al u^\nu (\partial_{\mu \si}u_\nu+\partial_{\si\nu}u_\mu)- \sv\frac{\sqrt{C^\rho C_\rho}}{\theta r}  (\partial_{\mu}u_\nu+\partial_{\nu}u_\mu)\pi^{\al\mu}u^\si \partial_{\al\si} u^\nu\\
& + B(\partial^2 g, \partial s, \partial C, \partial u)=0.
\end{split}
\ee
We note that in the derivation of $B$ in \eqref{s1} at first one obtains derivatives in $\theta$ and $r$, which get replaced by $\partial s$, and $\partial F$.  Then in view of the comment at the beginning of this section, $\partial F$ gets replaced by  $\partial C$ and $\partial g$.

\subsection{Equation for $u$.}

The derivation of the equation for $u$ is long, requiring several
calculations. We shall break them into short claims in order to 
facilitate the reading.

Since $F >0$, inspired by \cite{LichnerowiczBookGR} we can define the conformal metric
\begin{gather}
\overline{g} = F^2 g,
\nonumber
\end{gather}
and denote by $\overline{\nabla}$ covariant differentiation in the $\overline{g}$-metric.
We also let
\begin{gather}
\overline{C}^\al = F^{-1} u^\al,
\label{C_bar_inv}
\end{gather}
i.e., $\overline{C}^\al$ is $C_\al$ with index raised in the $\overline{g}$ metric, so
that 
\begin{gather}
\overline{C}^\al C_\al = 1.
\label{unit_C_bar}
\end{gather}
It also follows that
\begin{gather}
\Om_{\al\be} = \overline{\nabla}_\al C_\be - \overline{\nabla}_\be C_\alpha.
\label{vorticity_nabla_bar}
\end{gather}
If $v$ is a one-form, a direct calculation gives
\begin{gather}
\overline{\nabla}_\al v_\be = 
\nabla_\al v_\be - K_\al v_\be - K_\be v_\al + K^\rho v_\rho g_{\al\be},
\label{relation_cov_der}
\end{gather}
where $K_\al = \partial_\al \log F = \frac{\partial_\al F}{F}$.
The following standard identities will also be needed,
\begin{gather}
\nabla_\al \nabla_\be v^\la - \nabla_\be \nabla_\al v^\la 
= R_{\al\be \mss \ga}^{\msb \la} v^\ga ,
\nonumber
\end{gather}
from which it follows
\begin{gather}
\nabla_\al \nabla_\be v^\al - \nabla_\be \nabla_\al v^\al  
= R_{\be\ga} v^\ga.
\label{contracted_Ricci_id}
\end{gather}

To derive the equation for $u$ we need to compute the divergence of $T_{\al\be}$.
For this it will be convenient to set
\begin{gather}
\Si_{\al\be} =
\pi_\al^\mu  \pi_\be^\nu (\nabla_\mu  C_\nu + \nabla_\nu C_\mu ),
\label{shear_C}
\end{gather}
which can be written as
\begin{gather}
\Si_{\al\be} =
F \pi_\al^\mu  \pi_\be^\nu (\nabla_\mu  u_\nu + \nabla_\nu u_\mu ),
\label{shear_u}
\end{gather}
since  $u^\be\pi^\mu_\beta=0$. $\Si_{\al\be}$ is sometimes called the shear tensor.


\begin{claim}
Let
\begin{gather}
\overline{\Si}_{\al\be} = \overline{\nabla}_\al C_\be + 
\overline{\nabla}_\be C_\al - \overline{C}^\la( \overline{\nabla}_\la C_\al C_\be +
\overline{\nabla}_\la C_\be C_\al ).
\nonumber
\end{gather}
Then the following two identities hold:
\begin{gather}
\overline{\Si}_{\al\be} = \Si_{\al\be} +  2 \pi_{\al\be} u^\rho \partial_\rho F,
\nonumber
\end{gather}
and
\begin{gather}
\overline{\Si}_{\al\be} = 2 \overline{\nabla}_\be C_\al + \Theta_{\al\be},
\nonumber
\end{gather}
where
\begin{gather}
\Theta_{\al\be} = \Om_{\al\be} - u^\la(\Om_{\la \al} u_\be + \Om_{\la  \be} u_\al ).
\nonumber
\end{gather}
\label{claim_1}
\end{claim}
\begin{proof}
Using (\ref{relation_cov_der}), (\ref{dyn_vel}), (\ref{C_bar_inv}), a long
but not difficult calculation gives
\begin{gather}
\overline{\Si}_{\al\be} =
F(\nabla_\al u_\be + \nabla_\be u_\al - u_\be u^\nu \nabla_\nu u_\al 
- u_\al u^\mu \nabla_\mu 
u_\be) + 2 \pi_{\al\be} u^\rho \partial_\rho F.
\nonumber
\end{gather}
But in light of (\ref{useful}) we have
\begin{gather}
\pi_\al^\mu \pi_\be^\nu(\nabla_\mu u_\nu + \nabla_\nu u_\mu) 
= 
\nabla_\al u_\be + \nabla_\be u_\al - u_\be u^\nu \nabla_\nu u_\al 
- u_\al u^\mu \nabla_\mu 
u_\be ,
\nonumber
\end{gather}
so that (\ref{shear_u}) gives
\begin{gather}
\overline{\Si}_{\al\be} = \Si_{\al\be} +  2 \pi_{\al\be} u^\rho \partial_\rho F.
\nonumber
\end{gather}
For the second identity, use (\ref{vorticity_nabla_bar}) to get
\begin{align}
\begin{split}
\overline{\Si}_{\al\be} & = \overline{\nabla}_\al C_\be + 
\overline{\nabla}_\be C_\al - \overline{C}^\la( \overline{\nabla}_\la C_\al C_\be +
\overline{\nabla}_\la C_\be C_\al ) \\
& = \Om_{\al \be }+ 
2 \overline{\nabla}_\be C_\al - \overline{C}^\la( \Om_{\la\al} C_\be +
\Om_{\la\be} C_\al ) - \overline{C}^\la( \overline{\nabla}_\al C_\la C_\be + 
\overline{\nabla}_\be C_\la C_\al ).
\end{split}
\nonumber
\end{align}
The result now follows by noticing that 
(\ref{unit_C_bar}) gives $\overline{C}^\la \overline{\nabla}_\al C_\la = 0 
= \overline{C}^\la  \overline{\nabla}_\be  C_\la$ and using (\ref{dyn_vel}).
\nonumber
\end{proof}

\begin{claim}
We have
\begin{gather}
\pi^{\ga\be} \nabla_\al \overline{\Si}^\al_{\mss \be } =
- 2 \pi^{\ga\be} K^\al \Om_{\al\be} + \pi^{\ga\be} \nabla_\al \Theta^\al_{\mss \be}. 
\nonumber
\end{gather}
\label{claim_2}
\end{claim}
\begin{proof}
Using claim \ref{claim_1},  (\ref{relation_cov_der}), (\ref{contracted_Ricci_id}), 
and (\ref{original_incompressible}), 
one gets
\begin{align}
\begin{split}
\nabla_\al \overline{\Si}^\al_{\mss \be } & =
2 \nabla_\al \nabla_\be C^\al - 2 \nabla_\al K^\al C_\be - 2 K^\al \Om_{\al \be }
+ \nabla_\al \Theta^\al_{\mss \be} \\
& = 2 R_{\al\be}C^\al - 2 \nabla_\al K^\al C_\be - 2 K^\al \Om_{\al \be }
+ \nabla_\al \Theta^\al_{\mss \be}.
\end{split}
\nonumber
\end{align}
Now, from (\ref{ens1}), the fact that $\pi_{\al\be} C^\al = F \pi_{\al\be}u^\al = 0$,
and the form of $T_{\al\be}$, it follows that
\begin{gather}
R_{\al\be}C^\al = (rF - p -\frac{1}{2}T)C_\be,
\nonumber
\end{gather}
so that 
\begin{align}
\begin{split}
\nabla_\al \overline{\Si}^\al_{\mss \be } & =
(2rF - 2p - T - 2 \nabla_\al K^\al )C_\be   - 2 K^\al \Om_{\al \be }
+ \nabla_\al \Theta^\al_{\mss \be},
\end{split}
\nonumber
\end{align}
from which the claim follows upon contracting with $\pi^{\ga\be} $ and using again 
$\pi_{\al\be} C^\al = 0$
\end{proof}

\begin{claim} We have
\begin{gather}
2 u^\al \pi^{\ga\rho}\nabla_\al \partial_\rho F =
- 2 \pi^{\ga\be} K^\al \Om_{\al\be} + \pi^{\ga\be}\nabla_\al\Theta^\al_{\mss\be}
+ B^\ga(\partial g, \partial s, \partial F, \partial u).
\nonumber
\end{gather}
\label{claim_3}
\end{claim}
\begin{proof}
Combining the first identity of claim \ref{claim_1} with claim \ref{claim_2},
\begin{gather}
\pi^{\ga\be} \nabla_\al \Si^\al_{\mss \be } =
-2u^\rho \pi^{\ga\al}\nabla_\al \partial_\rho F 
- 2 \pi^{\ga\be} K^\al \Om_{\al\be} + \pi^{\ga\be} \nabla_\al \Theta^\al_{\mss \be}+ B^\gamma(\partial g, \partial F, \partial u). 
\nonumber
\end{gather}
Writing (\ref{T_vis_i}) as $T_{\al\be} = \widehat{T}_{\al\be} + \sv \Si_{\al\be}$,
 noticing that $ \pi^{\ga\be} \nabla_\al  \widehat{T}_{\al\be} = B^\ga(\partial g,
\partial F, \partial u)$, and invoking (\ref{original_conservation}),
we have
\begin{gather}
\pi^{\ga\be} \nabla_\al \Si^\al_{\mss \be } = 
B^\ga(\partial g, \partial s,
\partial F, \partial u),
\label{projection_div_shear_lower_order}
\end{gather}
since $\sv > 0$. The claim follows from these last two equalities. 
\end{proof}

\begin{claim} We have
\begin{gather}
u^\al \nabla_\al u_\be = \pi^\al_\be \frac{\partial_\al F}{F} + \frac{1}{F} u^\al 
\Om_{\al\be}.
\nonumber
\end{gather}
\label{claim_4}
\end{claim}
\begin{proof}
From (\ref{relation_cov_der}) and (\ref{dyn_vel}),
\begin{gather}
\overline{\nabla}_\al C_\be =  F \nabla_\al u_\be  - \partial_\be F u_\alpha+
u^\rho \partial_\rho F g_{\al\be},
\nonumber
\end{gather}
which upon contraction gives
\begin{gather}
\overline{C}^\al \overline{\nabla}_\al C_\be = u^\al \nabla_\al u_\be 
- \pi^\al_\be \frac{\partial_\al F}{F}.
\nonumber
\end{gather}
Contracting (\ref{vorticity_nabla_bar}) with $\overline{C}^\al$, using
(\ref{unit_C_bar}) and the above equality, one obtains the result, after rewriting
$C$ in terms of $F$ and $u$.
\end{proof}

\begin{claim} We have
\begin{align}
\begin{split}
2 \pi^{\ga\be} \nabla^\al \Si_{\al\be} & =
2 F \nabla^\al \nabla_\al u^\ga 
- \pi^{\ga\be}\nabla_\al \Theta^\al_{\mss \be} - 2u^\al u^\rho \nabla_\al \Om_\rho^{\mss \ga}
\\
& + 2 \pi^{\al\mu} \pi^{\ga \nu} \nabla_\al \nabla_\nu C_\mu+ 
B^\ga(\partial^2 g, \partial  s, \partial F, \partial u, \Om, \partial C).
\end{split}
\nonumber
\end{align}
\label{claim_5}
\end{claim}
\begin{proof}
Use (\ref{dyn_vel}) in the first term on the right hand side of (\ref{shear_C}) to 
write it as
\begin{gather}
\Si_{\al\be} = F \pi_\al^\mu \pi_\be^\nu \nabla_\mu u_\nu + 
\pi_\al^\mu \pi_\be^\nu \nabla_\nu C_\mu ,
\nonumber
\end{gather}
where $\pi_\be^\nu u_\nu = 0$ has been employed. Applying 
$\pi^{\ga\be} \nabla^\al$, we get
\begin{gather}
\pi^{\ga\be} \nabla^\al \Si_{\al\be}
= F(\nabla_\al \nabla^\al u^\ga - u^\al u^\mu \nabla_\al \nabla_\mu u^\ga)
+ \pi^{\al\mu} \pi^{\ga \nu} \nabla_\al \nabla_\nu C_\mu
+ B^\ga (\partial g, \partial F, \partial u, \partial C),
\label{claim_5_exp_1}
\end{gather}
where we have used that (\ref{useful}) implies
\begin{gather}
u^\nu \nabla^\al \nabla_\al u_\nu 
=  \nabla^\al(u^\nu \nabla_\al u_\nu) - \nabla^\al u^\nu \nabla_\al u_\nu  
= B^\ga(\partial g, \partial u),
\nonumber
\end{gather}
and
\begin{gather}
u^\nu \nabla_\al \nabla_\mu u_\nu 
=  \nabla_\al(u^\nu \nabla_\mu  u_\nu) - \nabla_\al u^\nu \nabla_\mu  u_\nu  
= B^\ga(\partial g, \partial u).
\nonumber
\end{gather}
Commuting $u^\mu$ and $\nabla_\al$ 
one obtains, in light of claim \ref{claim_4},
\begin{gather}
F u^\al u^\mu \nabla_\al \nabla_\mu u^\ga 
= u^\al\pi^{\ga\rho} \nabla_\al \partial_\rho F + u^\rho u^\al \nabla_\al \Om_\rho^{\mss \ga}
+ B^\ga(\partial g, \partial F, \partial u, \Om),
\nonumber
\end{gather}
so that (\ref{claim_5_exp_1}) becomes
\begin{align}
\begin{split}
\pi^{\ga\be} \nabla^\al \Si_{\al\be}
& = F\nabla_\al \nabla^\al u^\ga - 
u^\al\pi^{\ga\rho} \nabla_\al \partial_\rho F - u^\rho u^\al \nabla_\al \Om_\rho^{\mss \ga}
\\
& + \pi^{\al\mu} \pi^{\ga \nu} \nabla_\al \nabla_\nu C_\mu
+ B^\ga (\partial g, \partial F, \partial u, \partial C, \Om).
\end{split}
\nonumber
\end{align}
The result now follows by using claim \ref{claim_3} to eliminate
$u^\al\pi^{\ga\rho} \nabla_\al \partial_\rho F$ from the above expression, 
after noticing that $K^\al \Om_{\al\be}$ can be absorbed into
$B^\ga$.
\end{proof}
In light of (\ref{projection_div_shear_lower_order}) and using the definition
of $\Theta$,  
claim \ref{claim_5} gives the desired equation for $u$, namely,
\begin{align}
\begin{split}
 g^{\mu \nu } \partial_{\mu \nu}  u_\ga &
 - \frac{1}{2 \sqrt{C^\rho C_\rho} }
 g^{\mu \nu} \partial_\nu \Om_{\mu \ga} 
- \frac{1}{2 \sqrt{C^\rho C_\rho} }
u^\mu u^\nu \partial_\mu\Om_{\nu \ga} 
+
\frac{1}{2 \sqrt{C^\rho C_\rho} }
 u_\ga u^\mu g^{\nu \be} \partial_\be \Om_{ \nu \mu} 
 \\
 &
+  
\frac{1}{ \sqrt{C^\rho C_\rho} } \pi^{\al\mu} \pi^\nu_\ga \partial_{\al\nu} C_\mu + 
B_\ga(\partial^2 g,  \partial  s, \partial u, \Om, \partial C) = 0,
\end{split}
\label{u_equation}
\end{align}
where we used $F>0$, and subsequently (\ref{dyn_vel}) to 
eliminate 
the $F$ dependence.

\subsection{Equations for $\Om$.}
Recalling that 
\begin{gather}
\cL_C \Om_{\al \be}
= C^\mu \nabla_\mu \Om_{\al\be} + \nabla_\al C^\mu \Om_{\mu \be}
+ \nabla_\be C^\mu \Om_{\al \mu},
\nonumber
\end{gather}
we see that (\ref{evolution_vorticity}) has the form
\begin{gather}
C^\mu \partial_\mu \Om_{\al \be} + q u^\mu \partial_{\mu\al} C_\be 
-qu^\mu \partial_{\mu \be} C_\al + B_{\al\be}(\partial^2 g, \partial s, \partial u, \Om, \partial C) = 0.
\label{Omega_equation}
\end{gather}

\subsection{Equations for $C$.}

In order to close the system, we need to specify equations of motion for 
$C$. Since all the equations from (\ref{original}) have already 
been employed above in the derivation of equations for $g$, $s$,  $u$,
and $\Om$,
one suspects that equations for $C$  should be determined by
some extra conditions, not explicitly present in (\ref{original}).
However, in order to do so without changing the content of the original
problem, one must choose equations that are necessarily satisfied
by any solution to (\ref{original}). Thus, convenient identities,
that follow from standard tensor calculus and our basic definitions,
will be employed.

Using the definition of the Hodge-Laplacian gives
\begin{gather}
\Delta C = d \de C + \de d C = \de \Om,
\nonumber
\end{gather}
where $\de C = - \nabla^\mu C_\mu = 0$ (see (\ref{original_incompressible})) and
the definition of $\Om$, (\ref{vorticity}), have been used. On the other hand,
recalling
\begin{gather}
(\Delta C)_\al = - \nabla^\mu \nabla_\mu C_\al + R_{\mu\al}C^\mu,
\nonumber
\end{gather}
we obtain the following equation for $C$:
\begin{gather}
g^{\mu\nu}\partial_{\mu\nu} C_\al - g^{\mu\nu} \partial_\mu \Om_{\nu \al}
+ B_\al(\partial^2 g, \Om, \partial C) = 0.
\label{C_equation}
\end{gather}

\subsection{The full system.}
The sought new system of equations consists of (\ref{ens5}),
(\ref{s1}), (\ref{u_equation}), 
(\ref{Omega_equation}), and (\ref{C_equation}).
These are 25 equations for the 25 unknowns: ten $g_{\al\be}$, one $s$,
four $u_\al$, six $\Om_{\al\be}$, and four $C_\al$.
We shall refer to this system as the modified
incompressible Einstein-Navier-Stokes system.
One important aspect of our proof consists in showing that  the 
modified incompressible Einstein-Navier-Stokes system
 forms a Leray-Ohya system \cite{LerayOhyaNonlinear}, which
  depends, among other things, on a counting  of 
derivatives.
For this purpose, it is convenient to write the system 
symbolically
as

\begin{align}
\left\{
\begin{matrix}
a_{11}(g) \partial^2 g & + & 0 & + & a_{13}(g,u,C)\partial u & + & 0  
\\
0  & + & a_{22}(g,u)\partial^2 s & + & a_{23}(g,s,\partial u, C) \partial^2 u   
&+& 0
 \\
0  & + & 0 & + & a_{33}(g) \partial^2 u  & + & a_{34}(g,u,C)\partial \Om  
\\
0  & + & 0 & + &0  & + & a_{44}(g,C)\partial \Om 
\\ 
0  & + & 0 & + &0  & + & a_{54}(g)\partial \Om  
\end{matrix}
\nonumber
\right.
\end{align}
\begin{subequations}
\begin{alignat}{4}
 & + && \hspace{1.2cm} 0 &&  = && \hspace{0.25cm} B_g( \partial g, s, u, C),  
 \label{eq1} \\
 &+ && \hspace{1.2cm} 0 &&  =  && \hspace{0.25cm} B_s(\partial^2 g, \partial s, \partial u, \partial C), 
  \label{eq2} \\
 & + && \hspace{0.25cm} a_{35}(g,u,C) \partial^2 C && = && 
 \hspace{0.25cm} B_u(\partial^2 g, \partial s, \partial u, \Om, \partial C), 
 \label{eq3} \\
 & + && \hspace{0.25cm} a_{45}(g,u) \partial^2 C && = && 
 \hspace{0.25cm} B_\Om(\partial^2 g, \partial s, \partial u, \Om, \partial C), 
 \label{eq4} \\ 
 & + && \hspace{0.25cm}  a_{55}(g) \partial^2 C && = && 
 \label{eq5} \hspace{0.25cm} B_C (\partial^2 g, \Om, \partial C) .
\end{alignat}
\label{new_system_symbolic}
\end{subequations}
We write this more succinctly as
\begin{gather}
A(V,\partial)V = B(V),
\nonumber
\end{gather}
where $V=(g,s,u,\Om, C)$, $B(V)=(B_g, B_s, B_u, B_\Om, B_C)$,
 and
\begin{align}
A(V,\partial) = \left(
\begin{matrix}
a_{11}(g) \partial^2  &  0 &  a_{13}(g,u,C)\partial  &  0  &  0  
 \\
0  &  a_{22}(g,u)\partial^2  &  a_{23}(g,s,\partial u, C) \partial^2   &  0  &  0 &
 \\
0  &  0 &  a_{33}(g) \partial^2   &  a_{34}(g,u,C)\partial &  a_{35}(g,u,C) \partial^2    \\
0  &  0 & 0  &  a_{44}(g,C)\partial   &  a_{45}(g, u) \partial^2   \\ 
0  &  0 & 0  &  a_{54}(g)\partial   &  a_{55}(g) \partial^2  
\end{matrix}
\right)
\label{matrix}
\end{align}

\section{Proof of Theorem \ref{main_theorem}. \label{main_proof}}
The main tool in the proof of theorem \ref{main_theorem} is 
the theory of weakly hyperbolic equations in Gevrey spaces developed
by Leray and Ohya
\cite{LerayOhyaLinear,LerayOhyaNonlinear, OhyaLinear}
and extended by Choquet-Bruhat \cite{CB_diagonal} to the type of non-diagonal
systems which will be of interest here. We shall not review these
constructions, except for those aspects that will be necessary
to fix our notation and conventions, referrring the reader to the 
above references for the complete discussion. 
Some aspects of the Leray-Ohya theory can
also be found (without proofs) in 
\cite{ChoquetBruhatGRBook, ChoquetetallAnalysisManifolds, DisconziViscousFluidsNonlinearity}.

For the reader's convenience, we start by recalling the definition of 
Gevrey spaces. As our proof is essentially local, with a global (in space)
solution obtained by a gluing argument, it suffices 
to give the definition in the case of $\RR^{n+1}$, whose coordinates
we denote by $(x_0, \dots, x_{n})$.
For a number $|X|  > 0$, let $X$ be the strip 
$0 \leq x^0  \leq |X|. $
The Gevrey space $\ga^{m, (\si)}(X)$ is defined as follows.
Let $S_t = \{ x^0 =t \}$, and 
\begin{gather}
|D^k u |_t = c(n,k) \sup_{|\al| \leq k} \p D^\al u \p_{L^2(S_t)},
\nonumber
\end{gather}
where 
$c(n,k)$ is a normalization constant. Then, for $\si \geq 1$,  and
$m$ a non-negative integer,
$ u \in \ga^{m, (\si)}(X)$  means that $u \in C^\infty(X)$, and
\begin{gather}
\sup_{|\be| \leq m, \, \al,  \, 0 \leq t \leq |X|} \frac{1}{ (1 + |\al| )^\si } 
\left( |D^{\be + \al} u |_t \right)^\frac{1}{ 1 + |\al|  } < \infty,
 \nonumber
 \end{gather}
where the sup over $\al$ is taken over multi-indices such that $\al_0 = 0$.
Analogously one defines such spaces for open sets, product spaces, etc. 
Intuitively, $\ga^{m, (\si)}(X)$ can be thought of as a space between
analytic and smooth functions, in the following sense. An analytic function
$u$ 
on $S_t$ obeys, on each compact set, an inequality of the form
$|D^\al u| \leq C^{|\al|+1} \al!$, for some $C>0$. Gevrey functions
of class $\si \geq 1$ are  smooth functions obeying
 the weaker inequality $|D^\al u| \leq C^{|\al|+1} (\al!)^\si$.
 Then, $\ga^{m, (\si)}(X)$ consists of those functions whose
 derivatives up to order $m$, restricted to each time slice $S_t$,
 belong to the Gevrey space of class $\si$ --- except that it is convenient
 to characterize the Gevrey spaces of $S_t$ with the help of  an integral norm, 
 as done above. Gevrey spaces of functions defined  on a hypersurface 
 $\Si \subset X$ (say, on $\{ x^0 = 0 \}$), are defined in an analogous fashion 
 and denoted by $\ga^{(\si)}(\Si)$. These will be the spaces
 where initial data is prescribed (notice that there is no supremum over 
 $\beta$ in this case).
  Gevrey spaces are important in particular because 
 it is possible to establish in them the well-posedness of certain PDEs that
 are known not to be well-posed in Sobolev or smooth spaces \cite{MizohataCauchyProblem}.
At the same time, Gevrey spaces allow constructions with compactly supported functions, 
an important tool in analysis not possible in the class of analytic 
functions. 
 See \cite{LerayOhyaNonlinear,RodinoGevreyBook} for details on Gevrey spaces and their applications.

Consider a system of $N$ partial differential 
equations and $N$ unknowns in  $X = \RR^n \times [0,T]$, 
and
denote the unknown as 
 $V=(v^I)$, $I=1,\dots, N$. Suppose that the
system has the following quasi-linear structure:  it is possible 
to attach to each unknown $v^I$ a non-negative integer $m_I$,  and to
each equation a non-negative integer $n_J$, such that the system reads
\begin{gather}
h^J_I(\partial^{m_K - n_J -1} v^K, \partial^{m_I - n_J}) v^I
+ b^J(\partial^{m_K - n_J - 1} v^K) = 0, \, J=1, \dots, N.
\label{general_system}
\end{gather}
The notation here is similar to the 
one we used to write system (\ref{new_system_symbolic}), namely,
$h^J_I(\partial^{m_K - n_J -1} v^K, \partial^{m_I - n_J})$
is a homogeneous differential operator of order $m_I - n_J$ (which can
be zero), whose coefficients depend on at most
$m_K - n_J -1$ derivatives of $v^K$, $K=1,\dots N$. The remaining terms, 
$b^J(\partial^{m_K - n_J - 1} v^K)$, also depend on at most 
$m_K - n_J -1$ derivatives of $v^K$,  $K=1,\dots N$.

Recall that the characteristic polynomial of (\ref{general_system}) 
at $x \in X$ and for a given $V$
 is 
the polynomial in the co-tangent space $T^*_x X$,
$p(\xi)$, $\xi \in T_x^* X$,
of degree
\begin{gather}
\ell = \sum_{I=1}^N m_I - \sum_{J=1}^N n_J,
\nonumber
\end{gather}
given by the principal part (of order $\ell$) 
of the characteristic determinant of the system, $
\det(h^J_I(\xi))$. 

Consider the 
Cauchy problem for (\ref{general_system}), with Cauchy data given
on $X_0 = \RR^n  \times \{ t = 0\}$. Assume that for any $x \in X_0$,
and with $V$ taking the values of the Cauchy data on $X_0$,
the characteristic polynomial $p(\xi)$ is a product of $q$
hyperbolic polynomials of orders $\ell_q$,
\begin{gather}
p(\xi) = p_{1}(\xi) \cdots p_q(\xi).
\nonumber
\end{gather}
Suppose finally that 
\begin{gather}
\max_q \ell_q \geq \max_I m_I - \min_J n_J.
\nonumber
\end{gather}
Building on the techniques developed in \cite{LerayOhyaNonlinear}
Choquet-Bruhat proved \cite{CB_diagonal}
under the above conditions, 
the Cauchy problem for (\ref{general_system}) has a unique solution 
$V$ in the Gevrey space $\ga^{\operatorname{m},(\si)}(X^\prime)$, 
where $X^\prime = \RR^n \times [0,T^\prime]$, $T^\prime \leq T$,
for a suitable integer
 $\operatorname{m} \geq 0$, and $1 \leq \si <
  \si_0 = \frac{q}{q-1}$ (the case $q=1$, 
$\si_0 = \infty$, corresponds to solutions in Sobolev spaces). Furthermore,
the solution enjoys the domain of dependence or finite propagation speed
property, with the domain of dependence of a point $x \in X^\prime$
determined by the characteristic cone $\{ p(\xi) = 0 \}$ at $x$. \\

\noindent \emph{Proof of Theorem \ref{main_theorem}.} We shall verify 
that the  modified incompressible Einstein-Navier-Stokes system
 is of the form (\ref{general_system}) and satisfies all the 
conditions given in \cite{CB_diagonal} which we summarized 
 above.

Consider the unknown $V=(g,s,u,\Om, C) = (v^1, v^2, v^3, v^4, v^5)$
for the system (\ref{new_system_symbolic}). Naturally, it is understood
that each $v^I$ and each equation in (\ref{new_system_symbolic}) represent, 
respectively, a set of unknowns and a set of equations, but they can be 
grouped together since they are all of the same form. For instance,
for all the ten unknowns $g$, all the equations take the same form
(\ref{ens5}). We also remark that, as it is standard in the study of the
evolution problem in General Relativity, although $V$ and
(\ref{new_system_symbolic}) are defined in a local coordinate patch, we
rely on the aforementioned results \cite{CB_diagonal, LerayOhyaNonlinear}, 
given for $\RR^n \times [0,T]$ ($n=3$ in our case),
using the finite propagation speed and a standard gluing argument 
to construct global in space solutions (see below).

One verifies that (\ref{new_system_symbolic}) has the form (\ref{general_system})
upon choosing
\begin{align}
\begin{array}{ccccc}
m_1 = 3, & m_2 = 2, & m_3 = 2, & m_4 = 1, & m_5 = 2, \\
n_1 = 1, & n_2 = 0, & n_3 = 0, & n_4 = 0, & n_5 = 0,
\end{array}
\label{indices}
\end{align}
where $m_1 = m(v^1) \equiv m(g)$, $m_2 \equiv m(v^2) = m(s)$,
$m_3 = m(v^3) \equiv m(u)$, $m_4 = m(v^4) \equiv m(\Om)$, 
$m_5 = m(v^5) \equiv m(C)$,  
$n_1 = n(\text{equation } (\ref{eq1}) )
\equiv n(\text{equation } (\ref{ens5}) )$, 
 $n_2 = n(\text{equation } (\ref{eq2}) ) \equiv 
 n(\text{equation } (\ref{s1}) )$,
 $n_3 = n(\text{equation } (\ref{eq3}) ) \equiv 
 n(\text{equation } (\ref{u_equation}) )$, 
  $n_4 = n(\text{equation } (\ref{eq4}) ) \equiv 
 n(\text{equation } (\ref{Omega_equation}) )$,
  $n_5 = n(\text{equation } (\ref{eq5}) ) \equiv 
 n(\text{equation } (\ref{C_equation}) )$,
  and letting $h^J_I$ be the differential operator whose matrix $(h^J_J)$
is given by (\ref{matrix}). Indeed, we list below
for each equation $J$ in (\ref{new_system_symbolic}), the value of $n_J$,
the highest
  derivatives 
of the each unknown entering in the coefficients and the right-hand side of the equation, and the  difference 
$m_I - n_J$:
\begin{gather}
\text{eq. } (\ref{eq1}): 
n_1 = 1;
\partial g, s, u, C;
\begin{cases}
m(g) - n_1 \equiv m_1 - n_1 = 2,  \\
m(s) - n_1 \equiv m_2 - n_1 = 1,  \\
m(u) - n_1 \equiv m_3 - n_1 = 1,  \\
m(\Om) - n_1 \equiv m_4 - n_1 = 0,  \\
m(C) - n_1 \equiv m_5 - n_1 = 1,  
\end{cases}
\nonumber
\end{gather}

\begin{gather}
\text{eq. } (\ref{eq2}): 
n_2 = 0;
\partial^2 g, \partial s, \partial u, \partial C;
\begin{cases}
m(g) - n_2 \equiv m_1 - n_2 = 3,  \\
m(s) - n_2 \equiv m_2 - n_2 = 2,  \\
m(u) - n_2 \equiv m_3 - n_2 = 2,  \\
m(\Om) - n_2 \equiv m_4 - n_2 = 1,  \\
m(C) - n_2 \equiv m_5 - n_2 = 2,  
\end{cases}
\nonumber
\end{gather}

\begin{gather}
\text{eq. } (\ref{eq3}): 
n_3 = 0;
\partial^2 g, \partial s, \partial u, \Om, \partial C;
\begin{cases}
m(g) - n_3 \equiv m_1 - n_3 = 3,  \\
m(s) - n_3 \equiv m_2 - n_3 = 2,  \\
m(u) - n_3 \equiv m_3 - n_3 = 2,  \\
m(\Om) - n_3 \equiv m_4 - n_3 = 1,  \\
m(C) - n_3 \equiv m_5 - n_3 = 2,  
\end{cases}
\nonumber
\end{gather}

\begin{gather}
\text{eq. } (\ref{eq4}): 
n_4 = 0;
\partial^2 g, \partial s, \partial u, \Om, \partial C;
\begin{cases}
m(g) - n_4 \equiv m_1 - n_4 = 3,  \\
m(s) - n_4 \equiv m_2 - n_4 = 2,  \\
m(u) - n_4 \equiv m_3 - n_4 = 2,  \\
m(\Om) - n_4 \equiv m_4 - n_4 = 1,  \\
m(C) - n_4 \equiv m_5 - n_4 = 2,  
\end{cases}
\nonumber
\end{gather}

\begin{gather}
\text{eq. } (\ref{eq5}): 
n_5 = 0;
\partial^2 g, \Om, \partial C;
\begin{cases}
m(g) - n_5 \equiv m_1 - n_5 = 3,  \\
m(s) - n_5 \equiv m_2 - n_5 = 2,  \\
m(u) - n_5 \equiv m_3 - n_5 = 2,  \\
m(\Om) - n_5 \equiv m_4 - n_5 = 1,  \\
m(C) - n_5 \equiv m_5 - n_5 = 2.  
\end{cases}
\nonumber
\end{gather}

As described in  \cite{CB_diagonal, LerayOhyaNonlinear}, 
the Cauchy data for a system of the form (\ref{general_system}) consists
of the functions $v^I$, along with their derivatives up to order $m_J - 1$,
on the initial surface. The initial data is also required to satisfy
some compatibility conditions, which essentially come from requiring
that the equations are satisfied on the initial time slice when
they take the values of the initial data. In our case, we have
to further ensure that the initial data for the system (\ref{new_system_symbolic})
is compatible with solutions of the original set of equations, 
i.e., (\ref{original}) (written in harmonic coordinates), and with
 (\ref{dyn_vel}) and (\ref{vorticity}).
 
The derivation, out of the original 
initial data, of Cauchy data for (\ref{new_system_symbolic}),
such that the conditions of the last paragraph are satisfied,
is done in similar fashion as in \cite{DisconziViscousFluidsNonlinearity}, and
therefore we shall skip the details. In a nutshell, one uses the 
equations of motion to derive what the new initial data ought to be.
In fact, such a procedure is commonly used in General Relativity 
when solutions to Einstein equations are found via a different
set of equations 
\cite{ChoquetBruhatGRBook,
CB_York,
DisconziRemarksEinsteinEuler,
DisconziVamsiEinsteinEuler,
Fischer_Marsden_Einstein_FOSH_1972, Fri1, Fri2, Fri3, Fri4, Fri5, FriRenCauchy, Lichnerowicz_MHD_book}. We remark, for future reference,
that although we are treating $\Om$, $C$, and $u$ as independent variables,
and hence we do not know yet that (\ref{dyn_vel}) and (\ref{vorticity})
hold true, these relation are satisfied by the initial data by the way
they are derived; see \cite{DisconziViscousFluidsNonlinearity}. 

Next, we need to compute the characteristic determinant of the system, 
$\det A(V,\xi)$, where $\xi$ is a co-vector and  $A(V,\xi)$ the principal symbol in 
the
direction of $\xi$. From (\ref{matrix}),
\begin{gather}
\det A(V,\xi) = 
\det a_{11}(g, \xi) \det a_{22}(g,u,\xi) \det a_{33}(g,\xi)
\det \widetilde{A}(g,u,C,\xi),
\label{det_product}
\end{gather}
where $a_{ij}(\cdot,\xi)$ and $\widetilde{A}(g,u,C,\xi)$ are, 
respectively, the principal symbols of the differential 
operators $a_{ij}$ and
\begin{align}
\widetilde{A}(g,u,C,\partial) = \left(
\begin{matrix}
 a_{44}(g,C)\partial   &  a_{45}(g,u) \partial^2   \\ 
 a_{54}(g)\partial   &  a_{55}(g) \partial^2  
\end{matrix}
\right).
\nonumber
\end{align}
From  (\ref{ens5}),
(\ref{s1}), (\ref{u_equation}), we find,
\begin{gather}
\det a_{11}(g, \xi) = (\xi^\mu \xi_\mu)^{10},
\label{det_a_11}
\end{gather}
\begin{gather}
\det a_{22}(g, u, \xi) = (u^\mu \xi_\mu)^2,
\label{det_a_22}
\end{gather}
and
\begin{gather}
\det a_{33}(g, \xi) = (\xi^\nu \xi_\nu)^4,
\label{det_a_33}
\end{gather}
where as usual the indices are raised with $g$, i.e.,  
$\xi^\mu \xi_\mu = g^{\mu \nu} \xi_\nu \xi_\mu$, and
$u^\mu \xi_\mu = g^{\mu\nu} u_\nu \xi_\mu$. The powers
$10$ and $4$  in (\ref{det_a_11}) and
(\ref{det_a_33}) come, respectively, from the fact that 
(\ref{ens5}) corresponds to 
  ten equations  and (\ref{u_equation}) to
four equations, whereas the power $2$ 
in (\ref{det_a_22}) comes from the double characteristic
$u^\al u^\be \xi_\al \xi_\be$ of $u^\al u^\be \partial_{\al\be} s$
in (\ref{s1}). 

The operator $\widetilde{A}$ has a more complicated structure,
which requires us to be more explicit. Recalling that $\Om$
has six independent components, the components $(\Om,C)$ in 
$V=(g,s,u,\Om, C)$ are
\begin{gather}
(\Om_{01}, \Om_{02}, \Om_{03}, \Om_{12}, \Om_{13}, \Om_{23},
C_0, C_1, C_2, C_3),
\nonumber
\end{gather}
From (\ref{Omega_equation}) and (\ref{C_equation}), we then see 
that $\widetilde{A}$ has the following form
\begin{gather}
\begin{bmatrix}
C^\mu \partial_\mu & 0 & 0 & 0 & 0 & 0 & - q u^\mu \partial_{\mu 1} &
   q u^\mu \partial_{\mu 0} & 0 & 0 \\
0 & C^\mu \partial_\mu & 0 & 0 & 0 & 0 & -q u^\mu \partial_{\mu 2} &
   0 &  q u^\mu \partial_{\mu 0} &  0\\
0 & 0 & C^\mu \partial_\mu & 0 & 0 & 0 & -q u^\mu \partial_{\mu3} & 
   0 & 0 & q u^\mu \partial_{\mu 0 } \\
0 & 0 & 0 & C^\mu \partial_\mu & 0 & 0 & 0 & -q u^\mu \partial_{\mu 2} &
    q u^\mu \partial_{\mu 1}  & 0 \\
0 & 0 & 0 & 0 & C^\mu \partial_\mu & 0 & 0 & -q u^\mu \partial_{\mu 3} & 
    0 & q u^\mu \partial_{\mu1}  \\
0 & 0 & 0 & 0 & 0 & C^\mu \partial_\mu & 0 & 0 & -q u^\mu \partial_{\mu 3} & 
   qu^\mu \partial_{\mu 2}    \\
g^{\mu 1} \partial_\mu &  g^{\mu 2} \partial_\mu  & g^{\mu 3} \partial_\mu &
    0 & 0 & 0 & g^{\mu \nu } \partial_{\mu \nu} & 0 & 0 & 0 \\
- g^{\mu 0} \partial_\mu & 0 & 0 & g^{\mu 2} \partial_\mu & g^{\mu 3} \partial_\mu &
   0 & 0 & g^{\mu\nu} \partial_{\mu\nu} & 0 & 0 \\
0 & - g^{\mu 0} \partial_\mu & 0 & - g^{\mu 1} \partial_\mu & 0 &
      g^{\mu 3} \partial_\mu & 0 & 0 & g^{\mu\nu}\partial_{\mu\nu} & 0 \\
0 & 0 & - g^{\mu 0} \partial_\mu & 0 & - g^{\mu 1} \partial_\mu & 
     -g^{\mu 2} \partial_\mu & 0 & 0 & 0 & g^{\mu\nu} \partial_{\mu\nu}
\end{bmatrix}
\nonumber
\end{gather}
$\widetilde{A}(g,u,C,\xi)$ is given by
\begin{gather}
\begin{bmatrix}
C^\mu \xi_\mu & 0 & 0 & 0 & 0 & 0 & - q u^\mu \xi_{\mu} \xi_{ 1} &
   q u^\mu \xi_{\mu} \xi_{ 0} & 0 & 0 \\
0 & C^\mu \xi_\mu & 0 & 0 & 0 & 0 & - q u^\mu \xi_{\mu} \xi_{ 2} &
   0 &  q u^\mu \xi_{\mu} \xi_{ 0} &  0 \\
0 & 0 & C^\mu \xi_\mu & 0 & 0 & 0 & - q u^\mu \xi_{\mu} \xi_{ 3} & 
   0 & 0 & q u^\mu \xi_{\mu } \xi_{ 0 } \\
0 & 0 & 0 & C^\mu \xi_\mu & 0 & 0 & 0 & - q u^\mu \xi_{\mu} \xi_{ 2} &
    q u^\mu \xi_{\mu } \xi_{ 1}  & 0 \\
0 & 0 & 0 & 0 & C^\mu \xi_\mu & 0 & 0 & - q u^\mu \xi_{\mu} \xi_{ 3} & 
    0 & q u^\mu \xi_{\mu } \xi_{ 1}  \\
0 & 0 & 0 & 0 & 0 & C^\mu \xi_\mu & 0 & 0 & - q u^\mu \xi_{\mu} \xi_{ 3} & 
     q u^\mu \xi_{\mu} \xi_{ 2}   \\
\xi^1 &  \xi^2  & \xi^3 &
    0 & 0 & 0 & \xi^\mu \xi_\mu  & 0 & 0 & 0 \\
- \xi^0 & 0 & 0 & \xi^2 & \xi^3 &
   0 & 0 & \xi^\mu \xi_\mu  & 0 & 0 \\
0 & - \xi^0 & 0 & - \xi^1 & 0 &
      \xi^3 & 0 & 0 & \xi^\mu \xi_\mu & 0 \\
0 & 0 & - \xi^0 & 0 & - \xi^1 & 
     -\xi^2 & 0 & 0 & 0 & \xi^\mu \xi_\mu 
\end{bmatrix}
\nonumber
\end{gather}
The determinant of the above matrix can be computed, yielding, after much algebra,
\begin{gather}
F^3(F+q)^2 (u^\mu \xi_\mu)^6 (\xi^\la \xi_\la)^2 P(\xi),
\label{det_complicated}
\end{gather}
where 
\begin{gather}
P(\xi) = A \xi_0^4 + B \xi_0^2 + C
\nonumber
\end{gather}
and we have used that $g_{00} = 1$ and $g_{0i} = 0$. 
The coefficients $A$, $B$, and $C$ are given by
\begin{gather}
A = F+q,
\nonumber
\end{gather}
\begin{gather}
B = 2 F \xi_1 \xi^1 + 2q \xi_1 \xi^1 + 2F \xi_2 \xi^2 + 2q \xi_2 \xi^2 +  q \xi_3 \xi^2 + 2F \xi_3 \xi^3
+ q \xi_3 \xi^3,
\nonumber
\end{gather}
and
\begin{gather}
C = F (\xi_1 \xi^1)^2 + q (\xi_1 \xi^1)^2 + 2F \xi_1 \xi^2 \xi_2 \xi^2 + 2q \xi_1 \xi^2 \xi_2 \xi^2
+ q \xi_1 \xi_3 \xi^1 \xi^2 + F (\xi_2 \xi^2)^2 + q (\xi_2 \xi^2)^2
\nonumber \\
+ q \xi_2 \xi_3 (\xi_2)^2 + 2 F \xi_1 \xi_3 \xi^1 \xi^3 + q \xi_1 \xi_3 \xi^1 \xi^3
+ 2F \xi_2 \xi_3 \xi^2 \xi^3 + q \xi_2 \xi_3 \xi^2 \xi^3 + q (\xi_3)^2 \xi^2 \xi^3 + F (\xi_3 \xi^3)^2.
\nonumber
\end{gather}
We investigate the roots of $P(\xi)$. We have
\begin{gather}
(\xi_0)^2 = \frac{-B \pm \sqrt{ B^2 - 4AC}}{2A}.
\nonumber
\end{gather}
We need to verify that the right-hand side is real and non-negative. A long but not difficult 
computation reveals that
\begin{gather}
B^2 - 4AC = q^2 \xi_3^2 (\xi^2 - \xi^3)^2,
\nonumber
\end{gather}
assuring reality. For non-negativity, it suffices to show that $-B - \sqrt{ B^2 - 4AC} \geq 0$.
For this, assume first that we are working at a point where $g$ equals the Minkowski metric, so that,
after some more algebra and recalling our sign convention,
\begin{gather}
-B - \sqrt{ B^2 - 4AC} =
2(F+q)( (\xi_1)^2 + (\xi_2)^2 + \frac{(\xi_3)^2}{2} ) + F (\xi_3)^2  + q \xi_3 \xi_2
- \sqrt{ q^2 \xi_3^2 (\xi_2 - \xi_3)^2 }.
\nonumber
\end{gather}
It suffices to analyze the case where the term $q \xi_3 \xi_2$ gives a non-positive contribution. 
Thus we can replace $q \xi_3 \xi_2$ by $- q \xi_3 \xi_2$ and assume that $\xi_2 \geq 0$ and
$\xi_3 \geq 0$, in which case the above becomes
\begin{gather}
-B - \sqrt{ B^2 - 4AC} =
2(F+q)( (\xi_1)^2 + (\xi_2)^2 + \frac{(\xi_3)^2}{2} ) + F (\xi_3)^2  - q \xi_3 \xi_2
- q \xi_3 |\xi_2 - \xi_3 |.
\nonumber
\end{gather}
If $\xi_2 \geq \xi_3$, we find that
\begin{align}
\begin{split}
-B - \sqrt{ B^2 - 4AC} & =
2(F+q)( (\xi_1)^2 + (\xi_2)^2 + \frac{(\xi_3)^2}{2} ) + F (\xi_3)^2  - q \xi_3 \xi_2
- q \xi_3 (\xi_2 - \xi_3 ) \\
& = 2(F+q)( (\xi_1)^2 + (\xi_2)^2 + \frac{(\xi_3)^2}{2} ) + (F+q) (\xi_3)^2 - 2 q \xi_2 \xi_3
\\
& \geq 2(F+q)( (\xi_1)^2 + (\xi_2)^2 + \frac{(\xi_3)^2}{2} ) + (F+q) (\xi_3)^2 - 2 q (\xi_2)^2 \\
& = 2 (F+q) (\xi_1)^2 + 2F (\xi_2)^2 + 2 (F+q) (\xi_3)^2 \geq 0.
\end{split}
\label{inequality_Min_1}
\end{align}
If $\xi_2 \leq \xi_3$:
\begin{align}
\begin{split}
-B - \sqrt{ B^2 - 4AC} & =
2(F+q)( (\xi_1)^2 + (\xi_2)^2 + \frac{(\xi_3)^2}{2} ) + F (\xi_3)^2  - q \xi_3 \xi_2
- q \xi_3 (\xi_3 - \xi_2 ) \\
& = 2(F+q)( (\xi_1)^2 + (\xi_2)^2 )  + (F+q) (\xi_3)^2 + F (\xi_3)^2 - q (\xi_3)^2 \\
& = 2 (F+q) ( (\xi_1)^2 + (\xi_2)^2 ) + 2 F (\xi_3)^2 \geq 0.
\end{split}
\label{inequality_Min_2}
\end{align}
Therefore, we conclude that $P(\xi)$ factors as the product of two hyperbolic polynomials of degree two, 
$P(\xi) = P_1(\xi) P_2(\xi)$,
at least at a point where the metric equals the Minkowski metric.

Now we consider the general case, i.e., when $g$ does not necessarily equal the Minkowski metric.
Consider the initial hypersurface $\Si$ where the Cauchy data is given. We can assume that the coordinate
chart on $\Si$ is the neighborhood of a point $p$ such that $g_{\al\be}(p) = \eta_{\al\be}$, where
$\eta_{\al\be}$ is the Minkowski metric. Notice that we can still assure that the same coordinates
are harmonic coordinates, since the latter are determined by prescribing the first derivatives of $g$.
Since (\ref{inequality_Min_1}) and (\ref{inequality_Min_2}) are strict inequalities when $\xi \neq 0$,
we see that $-B - \sqrt{B^2 - 4AC} \geq 0$ for points sufficiently near $p$. Therefore, $P(\xi)$
is the product of two hyperbolic polynomial, as desired (notice that the reality condition previously
verified did not use $g_{\al\be}(p) = \eta_{\al\be}$).

Combining (\ref{det_product}), (\ref{det_a_11}), (\ref{det_a_22}),
(\ref{det_a_33}),  (\ref{det_complicated}), and the above,
we conclude that
\begin{gather}
\det A(V,\xi) = F^3 (F+ q)^3  (\xi^\mu \xi_\mu)^{14} 
(u^\nu \xi_\nu)^6  (u^\tau \xi_\tau \xi^\rho \xi_\rho )^2 P_1(\xi) P_2(\xi),
\label{product_hyp_pol}
\end{gather}
where it is understood that the above expression is evaluated
at the initial data since,
as mentioned in the beginning of this section, we need to verify
the hyperbolicity conditions of  \cite{CB_diagonal} when 
the unknown $V$ takes the values of the Cauchy data\footnote{More precisely,
we need to verify the hyperbolicity conditions when $V$ and its derivatives up to an order determined by the compatibility 
conditions take the values of the Cauchy data 
\cite{CB_diagonal, DisconziViscousFluidsNonlinearity}. Here, however, the 
coefficients appearing in the determinant (\ref{product_hyp_pol}) do not
involve derivatives of $V$.}. 
In particular, also as already mentioned,
even though we
are treating $C$ and $u$ as independent variables, for the initial data
it holds that $C = F u$, and we used this fact to eliminate $C$ from (\ref{product_hyp_pol}).

It is well-known (see e.g. \cite{Lichnerowicz_MHD_book}) that
the first, second, and third degree polynomials
$u^\tau \xi_\tau$, $\xi^\mu \xi_\mu$, and $u^\nu \xi_\nu \xi^\rho \xi_\rho$,
are hyperbolic as long as $g$ is a Lorentzian metric and $u$ is  time-like, conditions that are to be fulfilled when $V$ takes 
 our initial conditions. Also, $F+q > 0$ because $q > 0$ by hypothesis and 
 $F \geq 1$ by (\ref{indexF}), which holds for the initial data, 
 as well as the fact that $\ep \geq 0$, $p \geq 0$,
 and $r >0$.
$\det A(V,\xi)$ is, therefore, a product of 24 hyperbolic polynomials, 
with the highest degree of such polynomials being equal to three.
Thus, in the notation employed at the beginning of this section and
with indices $m_I$ and $n_J$ given by (\ref{indices}), we verify that
\begin{gather}
3 = \max_q \ell_q \geq \max_I m_I - \min_J n_J = 3 - 0 = 3.
\nonumber
\end{gather}
We also have $\si_0 = \frac{24}{24-1} = \frac{24}{23}$.

The coefficients of the the differential operator $A(V,\partial)$
depend polynomially on $V$, whereas  $B(V)$ is a rational function
of the functions $v^I$. The denominator of the rational expressions
appearing in $B(V)$ are products of $\sqrt{C^\rho C_\rho}$, 
$r = r(F,s) \equiv r(\sqrt{C^\rho C_\rho},s)$, and 
$\theta = \theta(F,s) \equiv \theta(\sqrt{C^\rho C_\rho},s)$.
Hence, recalling (\ref{temperature}) and that $F>0$,
the denominators in such rational expressions are, as functions of $V$, 
uniformly bounded away from zero (recall that $\Si$ is compact) 
when $V$ takes  the Cauchy data.

We have, therefore, verified all the conditions necessary to apply 
Choquet-Bruhat's theorem \cite{CB_diagonal}
combined with Leray and Ohya's results \cite{LerayOhyaNonlinear},
obtaining a short-in-time solution $V$ 
to (\ref{new_system_symbolic})
with $v^I \in \ga^{m_I,(\si)}(\Si \times [0,T])$, for $1 \leq \si < \si_0$
and some $T>0$.

It has to be shown that the solution $V$ to (\ref{new_system_symbolic})
yields a solution to the original set of equations (\ref{original}).
The argument to show this is very similar to the one employed in 
\cite{DisconziViscousFluidsNonlinearity, Lich_MHD_paper} (see also \cite{Lichnerowicz_MHD_book}), 
thus we just
mention the general idea. Consider the incompressible Einstein-Navier-Stokes
system written in harmonic coordinates. Pichon \cite{PichonViscous} has shown that this
system can be solved for analytic data (his work treated only the case of 
an equation of state that does not include entropy, but it is not
difficult to see that his procedure generalizes to the case of interest here).
By the way the Cauchy data for (\ref{new_system_symbolic}) is derived
out of the initial data for (\ref{original}), the analytic
solution to (\ref{original}) will satisfy the system (\ref{new_system_symbolic})
with $C_\al = F u_\al$ and $\Om_{\al\be} = \nabla_\al C_\be - \nabla_\be C_\al$.
For the case of initial data in Gevrey spaces, as in Theorem
\ref{main_theorem}, we approximate the initial data
by analytic Cauchy data, obtaining a sequence 
$\{ (g_j, u_j, r_j, s_j )\}$
of analytic solutions to (\ref{original}), and a corresponding 
sequence  $\{V_j\}$
of analytic solutions 
to (\ref{new_system_symbolic})
that converges to the solution $V$ obtained above. The estimates
on solutions derived by Leray and Ohya \cite{LerayOhyaNonlinear} assure that 
$\{ ( g_j, u_j, s_j, r_j )\}$ also converges to a limit
$\{ ( g, u,  s, r)\}$ that satisfies the incompressible Einstein-Navier-Stokes
system and belongs to the desired Gevrey class. It is well-known
that a solution to Einstein's equations in harmonic coordinates yields
a solution to the full system if and only if the constraint equations
are satisfied, which is the case by hypothesis\footnote{Here it is important to remind 
the reader of the discussion of section \ref{restriction_vorticity}. In particular, 
the solvability of the constraints under the further requirements imposed by (\ref{evolution_vorticity})
is at this point unknown.}. Finally, we
notice that $\pi^{\ga\be} \nabla^\al T_{\al\be} = 0$ implies
\begin{gather}
0 = u^\al u^\ga \nabla_\al u_\ga = \frac{1}{2} u^\al \partial_\al (u^\ga u_\ga),
\nonumber
\end{gather}
and therefore $u$, being unitary at time zero, remains unitary.

The existence of a domain of dependence also follows from the results
of \cite{LerayOhyaNonlinear}. The domain of dependence of the solution
is given by the intersection of the interior of the cones determined
by the hyperbolic polynomials appearing in the product
(\ref{product_hyp_pol}). All these cones have a common interior,
namely, the interior of the light-cone
$\xi^\mu \xi_\mu = g^{\mu\nu}\xi_\mu \xi_\nu \geq 0$. 
With the domain of dependence at hand, a standard gluing argument produces a
solution that is global in space and geometrically unique. This finishes
the proof of Theorem \ref{main_theorem}.

\vskip 0.2cm
\noindent \textbf{Acknowledgments.} We are grateful to the referees for reading the manuscript 
in much detail and making several important comments.

\bibliographystyle{plain}                                             
\bibliography{References.bib}

\end{document}